\documentclass[journal,comsoc]{IEEEtran}
\usepackage{setspace}
\usepackage{graphicx}
\usepackage{amssymb}
\usepackage{amsmath}
\usepackage{cite}
\usepackage{subfigure}
\usepackage{mathrsfs}
\usepackage[displaymath,mathlines]{lineno}
\usepackage{color}
\usepackage{tabulary}
\usepackage{multirow}
\usepackage{algpseudocode}
\usepackage{algorithm,algpseudocode}
\usepackage{algorithmicx}
\usepackage{pbox}
\usepackage{multicol}
\usepackage{lipsum}
\usepackage{amsthm}
\usepackage{relsize}
\usepackage{lipsum}
\usepackage{epstopdf}
\usepackage{mathtools}% http://ctan.org/pkg/mathtools
\usepackage{calc}% http://ctan.org/pkg/calc
\usepackage{makecell}
 \usepackage{url}

\makeatletter
\newcommand{\doublewidetilde}[1]{{%
		\mathpalette\double@widetilde{#1}}}
\newcommand{\double@widetilde}[2]{%
		\sbox\z@{$\m@th#1\widetilde{#2}$}%
		\ht\z@=.5\ht\z@
		\widetilde{\box\z@}}
\makeatother

\newtheorem{definition}{Definition}
\newtheorem{theorem}{Theorem}
\newtheorem{lemma}{Lemma}

\newtheorem{remark}{Remark}

\begin{document}
%
% paper title
% Titles are generally capitalized except for words such as a, an, and, as,
% at, but, by, for, in, nor, of, on, or, the, to and up, which are usually
% not capitalized unless they are the first or last word of the title.
% Linebreaks \\ can be used within to get better formatting as desired.
% Do not put math or special symbols in the title.
%\font\myfont=cmr4 at 12pt
\title{\huge Space-Terrestrial Cooperation Over Spatially Correlated Channels Relying on Imperfect Channel Estimates: Uplink Performance Analysis and Optimization}%\huge Satellite-Assisted Cell-Free Massive MIMO Communications: How Much Can We Gain?}

% author names and affiliations
% use a multiple column layout for up to three different
% affiliations
%\author{Author 1, Author 2, Author 3, Author 4, and Author 5
\author{\large Trinh~Van~Chien, \textit{Member}, \textit{IEEE}, Eva~Lagunas, \textit{Senior Member},  \textit{IEEE}, Tiep~M.~Hoang, \textit{Member},  \textit{IEEE}, Symeon~Chatzinotas, \textit{Fellow}, \textit{IEEE}, Bj\"{o}rn~Ottersten, \textit{Fellow}, \textit{IEEE}, and Lajos~Hanzo, \textit{Life Fellow}, \textit{IEEE} 
 \vspace*{-0.4cm}
 \thanks{Manuscript received xxx; revised xxx and xxx; accepted xxx. Date of publication xxx; date
 	of current version xxx. Parts of this paper was presented in the GLOBECOM 2022. This research is funded by Hanoi University of Science and Technology (HUST) under project number T2022-TT-001. This work has been supported by the Luxembourg National Research Fund (FNR) under the project MegaLEO (C20/IS/14767486). L. Hanzo would like to acknowledge the financial support of the
 	Engineering and Physical Sciences Research Council projects EP/W016605/1
 	and EP/P003990/1 (COALESCE) as well as of the European Research
 	Council's Advanced Fellow Grant QuantCom (Grant No. 789028). (Corresponding author: Trinh Van Chien.)}
\thanks{T.~V.~Chien is with the School of Information and Communication Technology (SoICT), Hanoi University of Science and Technology (HUST), Hanoi 100000, Vietnam (email: chientv@soict.hust.edu.vn).}
\thanks{E.~Lagunas, S.~Chatzinotas, and B.~Ottersten are with the Interdisciplinary Centre for Security, Reliability and Trust (SnT), University of Luxembourg, L-1855 Luxembourg, Luxembourg (email: eva.lagunas@uni.lu, symeon.chatzinotas@uni.lu, and bjorn.ottersten@uni.lu).}
\thanks{ T. M. Hoang is with the Department of Electrical Engineering, the University of Colorado Denver, Denver, CO 80204, USA (e-mail: minhtiep.hoang@ucdenver.edu).}
\thanks{L. Hanzo are with the School of Electronics and Computer Science, University of Southampton, Southampton SO17 1BJ, U.K. (e-mail: lh@ecs.soton.ac.uk).}
%}
}

% make the title area
\maketitle

% As a general rule, do not put math, special symbols or citations
% in the abstract
%\vspace*{-1.2cm}
\begin{abstract}
%\vspace*{-0.2cm}
A whole suite of innovative technologies and architectures have emerged in response to the rapid growth of wireless traffic. This paper studies an integrated network design that boosts system capacity through cooperation between wireless access points (APs) and a satellite for enhancing the network's spectral efficiency. \textit{As for our analytical contributions}, upon coherently combing the signals received by the central processing unit (CPU) from the users through the space and terrestrial links, we first mathematically derive an achievable throughput expression for the uplink (UL) data transmission over spatially correlated Rician channels. Our generic achievable throughput expression is applicable for arbitrary received signal detection techniques employed at the APs and the satellite under realistic imperfect channel estimates. A closed-form expression is then obtained for the ergodic UL data throughput, when maximum ratio combining is utilized for detecting the desired signals. \textit{As for our resource allocation contributions}, we formulate the max-min fairness and total transmit power optimization problems relying on the channel statistics for performing power allocation. The solution of each optimization problem is derived in form of a low-complexity iterative design, in which each data power variable is updated relying on a  closed-form expression. Our integrated hybrid network concept allows users to be served that may not otherwise be accommodated due to the excessive data demands. The algorithms proposed allow us to address the congestion issues appearing when at least one user is served at a rate below his/her target. The mathematical analysis is also illustrated with the aid of our numerical results that show the added benefits of considering the space links  in terms of improving the ergodic data throughput. Furthermore, the proposed algorithms smoothly circumvent any potential congestion, especially in face of high rate requirements and weak channel conditions.
%\vspace*{-0.3cm}
\end{abstract}
%\vspace*{-0.2cm}
\begin{IEEEkeywords}
%\vspace*{-0.2cm}
Cooperative network, space-terrestrial communications, linear processing, ergodic data throughput
%\vspace*{-0.2cm}
\end{IEEEkeywords}% no keywords
%\vspace*{-0.2cm}

% For peer review papers, you can put extra information on the cover
% page as needed:
% \ifCLASSOPTIONpeerreview
% \begin{center} \bfseries EDICS Category: 3-BBND \end{center}
% \fi
%
% For peerreview papers, this IEEEtran command inserts a page break and
% creates the second title. It will be ignored for other modes.
\IEEEpeerreviewmaketitle

%\vspace*{-0.2cm}
\section{Introduction}
%\vspace{-0.2cm}
The spectral vs energy efficiency trade-off of terrestrial communication systems has been remarkably improved in the recent decades, especially by the fifth-generation (5G) system \cite{9428708,lin20215g}. However, further escalation of the tele-traffic is anticipated with billions of devices managed by the  terrestrial wireless networks. The seamless tele-presence services of the near future require a high data rate and low end-to-end delays \cite{durisi2016toward}. To handle these increasing demands Massive MIMO (multiple-input multiple-output) techniques have been conceived \cite{ghazanfari2021model}, but the escalating inter-cell interference limits the performance of dense networks operating without any base-station (BS) collaboration. Cell-edge users may suffer from increased interference that leads to low  throughput.   

Future wireless systems will offer high throughput per user principally still based on the access to new spectrum, while intelligently coordinating a number of access points (APs) in a coverage area \cite{andrews2016we}. This leads to the concept of distributed MIMO systems \cite{pan2017joint} or Cell-Free Massive MIMO systems \cite{Ngo2015a}, which serve a group of users by a group of APs. The network then coherently combines different observations of the transmitted waves received over multiple heterogeneous propagation paths \cite{9500472}  using either maximum ratio combing (MRC) or minimum mean square error (MMSE) reception. 

Satellite communication has attracted renewed interest as a promising technique of providing services for many users across a large coverage area \cite{schwarz2019mimo}. However, given the large footprint of the satellite and its limited bandwidth shared by many users, both its area spectral efficiency and the per-user rate remain low. Geostationary (GEO) satellites have gained popularity, given their long coherence time \cite{bui2021robust,van2021user,an2018performance}. However, their  drawback is their excessive delay of about 120~ms and expensive manufacturing as well as launch. Consequently, the low latency, smaller size, and shorter delays of non-GEO (NGSO) satellites are considerable benefits \cite{kodheli2020satellite,leyva2021inter,pan2020performance}, especially extensions to megaconstellations  \cite{jia2021uplink}. Hence,  companies such as SpaceX,  OneWeb, TeleSAT, and Amazon have already started the deployment of large Low Earth Orbit (LEO) satellite constellations  \cite{pachler2021updated}.   

Both academia and industry have recently intensified their research of NGSO aided terrestrial communications \cite{3gpp2019study, abdelsadek2021future, riera2021enhancing}. The Digital Agenda for Europe initiative is also aiming for enhancing terrestrial connectivity \cite{niephaus2018toward}. As for the demands of tomorrow's networks, the authors of \cite{abdelsadek2021future} considered the performance of a space communication system that replaces terrestrial APs by LEO satellites. Despite integrating a LEO satellite into a terrestrial network \cite{riera2021enhancing}, the received signals were detected independently i.e., without exploiting the benefits of constructive received signal combination. As a further contribution, the coexistence of fixed satellite services and cellular networks was studied in \cite {du2018secure} for transmission over slow fading channels subject to individual user throughput constraints. The ergodic rate of the fast fading channels was considered in \cite{ruan2019energy, an2015performance} or the hybrid coverage probability and average interference modeling \cite{al2021modeling} under the assumption of perfect channel state information (CSI) and no spatial correlation. In a nutshell, the literature of  space-terrestrial integrated networks suffers from the following limitations: $i)$ most of the performance analysis and resource allocation studies rely on the idealized simplifying assumptions of perfect instantaneous CSI knowledge, which is challenging to acquire in practice, especially under high mobility scenarios; $ii)$ the spatial correlation between satellite antennas is ignored, even though it is unavoidable in the existing planar antenna arrays; $iii)$ all the users are treated equally, despite their heterogeneous throughput requirements and different channel conditions; and $iv)$ for a feasible solution, the networks are  supposed to satisfy all the user-specific throughput requirements, regardless of the finite network dimensions, which is a strong assumption in multiple access scenarios.

However, to the best of our knowledge, there is no analysis of the space-terrestrial network in the literature in the face of the spatial correlation imposed by an antenna array, when coherently combining the signal received from both the satellite and terrestrial APs. By taking advantage of both the distributed Cell-Free Massive MIMO structure and satellite communications, we evaluate the ergodic throughput of each user relying on a limited number of APs and demonstrate how the satellite enhances the system performance. Furthermore, we  study a pair of long-term power allocation problems relying on the knowledge of channel statistics. Explicitly, our main contributions are summarized as follows:
\begin{itemize}
\item[$i)$] We derive the achievable rate expression of each user in the uplink (UL) for transmission over spatially correlated fading channels, when relying on centralized data processing. If the MRC technique is used  both at the APs and at the satellite gateway, a closed-form expression of the ergodic net throughput will be derived. 
\item[$ii)$] Furthermore, we formulate a max-min fairness optimization problem that simultaneously allocates the  powers to all the scheduled users and guarantees uniform throughput for the entire network. In contrast to the interior-point method of \cite{ngo2017cell,pan2017joint}, we determine the user-specific optimal power for each user at a low complexity by exploiting the quasi-concavity of the objective function, the standard interference functions, and the bisection method.
\item[$iii)$] For supporting users who have different rate requirements, we formulate and solve a total transmit power minimization problem, while meeting the long-term individual throughput requirements. The proposed algorithms detect and handle any potential  congestion encountered, when the throughput requested by the users cannot be met.
\item[$iv)$] Our numerical results quantify the value of the satellite in improving both the total and the minimum user throughput. More explicitly, the users having poor channel conditions glean considerable benefits from the satellite's assistance. Besides, many users may still have their data throughput requirements satisfied in the face of congestion.
\end{itemize}

The rest of this paper is organized as follows: Section~\ref{Sec:SysModel} presents our  space-terrestrial communication system model, the channel model, and the UL channel estimation protocol. Our ergodic  throughput analysis is provided in Section~\ref{Sec:ULData}. Based on our closed-form   achievable rate expression, Section~\ref{Sec:PowerAllocation} formulates and solves our optimization problems under the constraints of limited power budgets and throughput requirements. Our numerical results are presented in Section~\ref{Sec:NumericalResults}, while our conclusions are offered in Section~\ref{Sec:Concl}. Table~\ref{Table:Notation} tabulates the common
	notation and symbols utilized throughout the paper.
\begin{table}[t]
\centering \caption{Notation and symbols} \label{Table:Notation}
%\resizebox{\textwidth}{!}{
		\begin{tabular}{ |c|c|c| } 
		\hline
		$(\cdot)^T$ & Regular transpose  \\
		$(\cdot)^H$ & Hermitian transpose  \\ 
		$\mathrm{tr}(\mathbf{X})$ & Trace of square matrix $\mathbf{X}$  \\ 
		$\mathbf{I}_N$ &  Identity matrix of size $N \times N$\\
		$\mathcal{CN}(\cdot, \cdot)$ &  Circularly symmetric Gaussian distribution \\
		 $\mathbb{E}\{\cdot\}$ & Expectation of a random variable \\
		$\mod(\cdot, \cdot)$ & Modulus operation \\
		$\lfloor \cdot \rfloor$  & Floor function \\
		$\otimes$ &  Kronecker product \\
		$|\mathcal{X}|$ & Cardinality of set $\mathcal{X}$ \\
		$J_1(\cdot)$  &  Bessel function of the first kind of order one \\
		 $\mathcal{O}(\cdot)$ & Big-$\mathcal{O}$ notation \\
 		\hline
	\end{tabular}
%}
%\vspace*{-0.8cm}
\end{table}
%\vspace{-0.5cm}
\section{System Model} \label{Sec:SysModel}
%\vspace{-0.2cm}
We consider a  distributed multi-user network comprising $M$ APs each equipped with a single receiver antenna (RA). The APs  cooperatively serve $K$ users in the UL, all equipped with a single transmit antenna (TA). The system performance is enhanced by the assistance of an NGSO satellite having $N$ RAs arranged in an $N_H \times N_V$-element rectangular array ($N =N_H \times N_V$), as illustrated in Fig.~\ref{FigSysModel}. Both the satellite gateway and the APs forward the UL signals received from the users to a central processing unit (CPU) by fronthaul links. As seen in Fig.~\ref{FigSysModel}, the APs rely on optical fronthaul links, while the satellite has a radio downlink (feeder link) to the ground station, which forwards the users' UL signal to the CPU. We assume that the optical fronthaul links and the feeder link has imperfect channel gains synthesized in a complex Gaussian distribution that influence both the pilot training and data transmission phases. Since the dispersive channels fluctuate both time and frequency over wideband systems, orthogonal frequency division multiplexing (OFDM) is used for mitigating it \cite{you2020massive}. A block-fading channel model is applied across the OFDM symbols, where the fading envelope is assumed to be frequency-flat through an entire OFDM symbol and then faded randomly for the next OFDM symbol. We assume that a fraction of $K$ subcarriers of each  OFDM symbol are known pilots, while the remaining $\tau_c - K$ subcarriers are used for UL payload data transmission. For our system model considered in Fig.~\ref{FigSysModel}, the satellite antenna's gain is sufficiently high to amplify the weak UL signals received from the distant terrestrial users on the ground \cite{3gpp2019study}. The following further assumptions are  exploited for pilot and data signal processing:
\begin{itemize}
	\item The UL channels are locally estimated both at the satellite gateway and at the APs to formulate the desired receiver combining vectors during the pilot-aided training phase. The detailed interpretation is presented in Section~\ref{Sec:ULP}.
	\item  In the UL data transmission phase, linear combining weights are applied to the signals received at the APs and separately to the satellite gateway before forwarding their linearly combined signals to the CPU for coherent receiver-combining. The detailed interpretation is presented in Section~\ref{Sec:UDT}.
\end{itemize}
Although the signal power received by the satellite may be significantly lower than that of the terrestrial APs, the large RA array of the satellite is capable of compensating this with the aid of its high receiver gain. Consequently,  the coherent receiver-combining applied at the CPU is still capable of improving the terrestrial links, provided that the satellite has a high TA gain and the ground station has a high RA gain for compensating the pathloss  \cite{perez2019signal}.
\begin{figure}[t]
	\centering
	\includegraphics[trim=2.6cm 2.6cm 2.2cm 1.7cm, clip=true, width=3.2in]{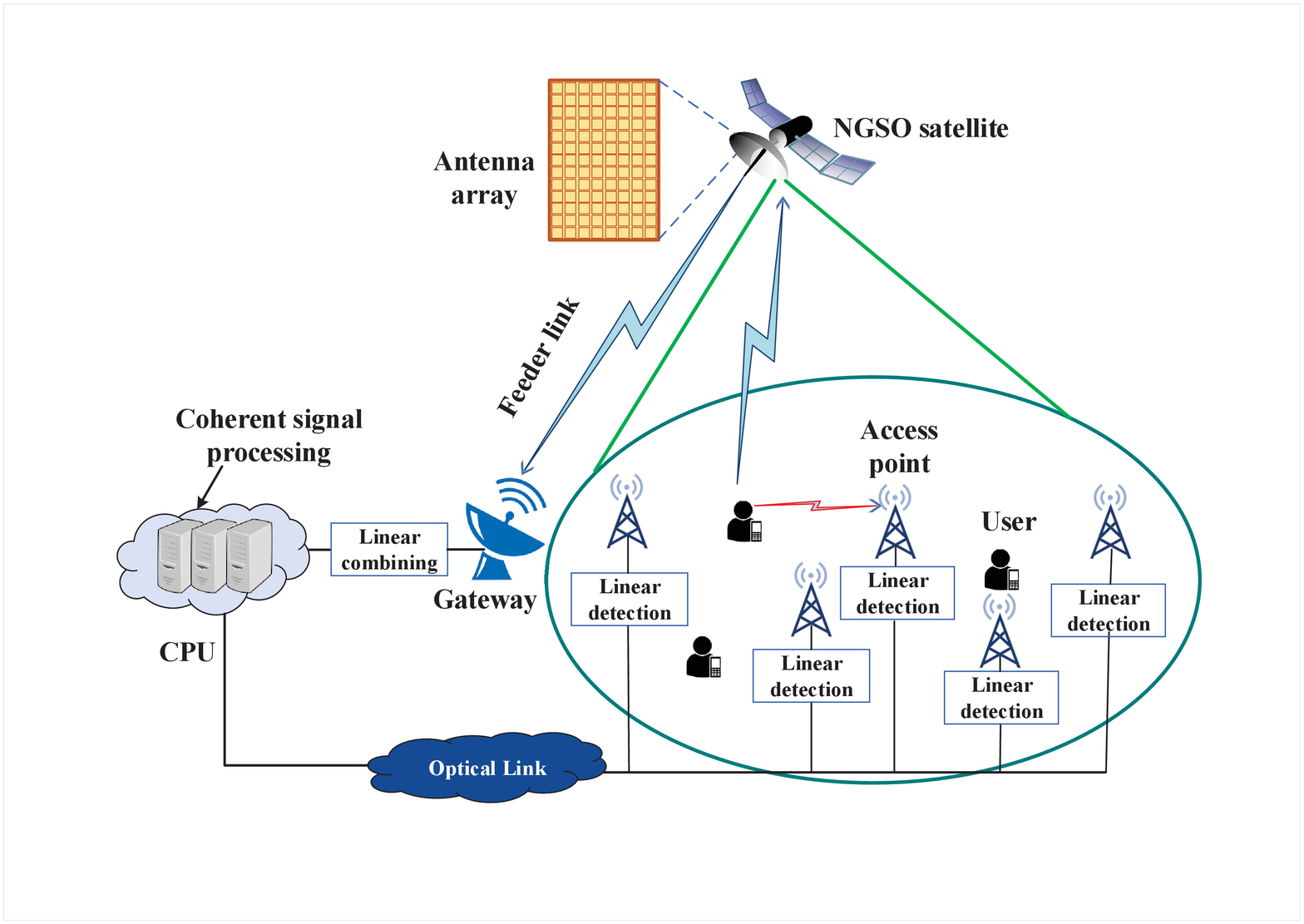} \vspace*{-0.2cm}
	\caption{Illustration of a cooperative satellite-terrestrial  wireless network.}
	\label{FigSysModel}
	%\vspace*{-0.9cm}
\end{figure}
%\vspace{-0.3cm}
\subsection{Channel Model} \label{Sec:ChannelModel}
%\vspace{-0.2cm}

The terrestrial UL channel between AP~$m$ and user~$k$, $\forall m,k,$ denoted by $g_{mk} \in \mathbb{C}$ is modeled as
%\begin{equation} \label{eq:channelmk}
$g_{mk}  \sim \mathcal{CN}(0, \beta_{mk})$,
%\end{equation}
where $\beta_{mk}$ is the large-scale fading coefficient that involves both the path loss and the shadow fading caused by obstacles. The channel between the UL transmitter of user~$k$ and the satellite receiver, denoted by $\mathbf{g}_k \in \mathbb{C}^N$, has been modeled according to the 3GPP recommendation (Release 15) \cite{3gpp2019study} and obeys the Rician distribution as
%\begin{equation} \label{eq:Channelgk}
$\mathbf{g}_k \sim \mathcal{CN}(\bar{\mathbf{g}}_k, \mathbf{R}_k)$,
%\end{equation}
where $\bar{\mathbf{g}}_k \in \mathbb{C}^N$ denotes the LoS components gleaned from the $N$ RAs in the UL. The matrix $\mathbf{R}_k \in \mathbb{C}^{N \times N}$ is the covariance matrix of the spatially correlated signals collected by the RAs of the satellite attenuated by the propagation loss.\footnote{The propagation loss in the carrier frequency range from $0.5$~GHz to $100$~GHz has been well documented in \cite{3gpp2019study}.} The LoS component is given by 
%\vspace*{-0.1cm}
\begin{equation} \label{eq:bargk}
\bar{\mathbf{g}}_k = \sqrt{\kappa_k \beta_k/(\kappa_k+1)} \left[e^{j \pmb{\ell}(\theta_k, \omega_k)^T \mathbf{c}_1}, \ldots,  e^{j \pmb{\ell}(\theta_k, \omega_k)^T \mathbf{c}_N}  \right]^T,
%\vspace*{-0.1cm}
\end{equation}
where $\theta_k$ and $\omega_k$ are the elevation and azimuth angle, respectively; $\kappa_k \geq 0$ represents the Rician factor;  and $\beta_k$ is the large-scale fading coefficient encountered between user~$k$ and the satellite, which depends both on the satellite's altitude and on the user's location (please see \eqref{eq:betak} in Section~\ref{Sec:NumericalResults} for a particular scenario). We assume that the antenna array is a rectangular surface (please see Fig.~\ref{FigSysModel}), whose wave form vector $\pmb{\ell}(\theta_k, \omega_k)$ \cite{massivemimobook,Chien2021TWC} is defined as 
%\vspace*{-0.1cm}
\begin{equation}
\pmb{\ell}(\theta_k, \omega_k) = \frac{2 \pi}{\lambda}[ \cos(\theta_k)\cos(\omega_k), \sin(\theta_k)\cos(\omega_k), \sin(\theta_k) ]^T,
%\vspace*{-0.1cm}
\end{equation}
with $\lambda$ being the wavelength of the carrier. In \eqref{eq:bargk}, there are $N$ indexing vectors, each given by
%\vspace*{-0.1cm}
\begin{equation}
\mathbf{c}_m = [0, \mathrm{mod}(n-1, N_H) d_H, \lfloor (n-1)/N_H \rfloor d_V ]^T,
%\vspace*{-0.1cm}
\end{equation}
where $d_H$ and $d_V$ represent the antenna spacing in the horizontal and vertical direction, respectively. The 3D channel model of \cite{ying2014kronecker,chatzinotas2009multicell} relies on the spatial correlation matrices of a planar antenna array, formulated as
%\vspace*{-0.1cm}
\begin{equation}
\mathbf{R}_k =  (\beta_k/(\kappa_k +1)) \mathbf{R}_{k,H} \otimes \mathbf{R}_{k,V},
%\vspace*{-0.1cm}
\end{equation}
where $\mathbf{R}_{k,H} \in \mathbb{C}^{N_H \times N_H}$ and $\mathbf{R}_{k,V} \in \mathbb{C}^{N_V \times N_V}$ are the spatial correlation matrices along the horizontal and vertical direction.
%\vspace*{-0.4cm}
\begin{remark}
The propagation channels considered in this paper involve several practical aspects \cite{LTE2017a,3gpp2019study}. The terrestrial channels represent isotropic environments having no dominant propagation path. Furthermore, the presence of a satellite generates extra paths associated with strong reflected waves. Hence, they are formulated for a compact antenna array  obeying a generic ray-based 3D channel model that splits the spatial correlation into the azimuth and elevation dimensions. This model is representative of isotropic scattering environments in the half-space in front of users. Our framework can be widely applied to different space-terrestrial communication scenarios by adopting the corresponding propagation settings of \cite{LTE2017a,3gpp2019study}.
%\vspace*{-0.3cm}
\end{remark}
%\vspace{-0.5cm}
\subsection{Uplink Pilot Training} \label{Sec:ULP}
%\vspace{-0.2cm}
 All the $K$ users simultaneously transmit their pilot signals in each coherence block of the UL. We assume to have the same number of orthogonal pilots as users, i.e., we have the set $\{ \pmb{\phi}_1, \ldots, \pmb{\phi}_K\}$, where the pilot $\pmb{\phi}_k \in \mathbb{C}^K$ is assigned to user~$k$ so that we have $\pmb{\phi}_k^H \pmb{\phi}_{k'} = 1 $ if $k=k'$. Otherwise, $\pmb{\phi}_k^H \pmb{\phi}_{k'} = 0 $.
The training signal received at AP~$m$, $\mathbf{y}_{pm} \in \mathbb{C}^K$, is a superposition of all the UL pilot signals transmitted over the propagation environment, which is formulated as
%\vspace*{-0.2cm}
\begin{equation} \label{eq:ypm}
	\mathbf{y}_{pm}^H = \sum\nolimits_{k=1}^K \sqrt{pK} g_{mk} \pmb{\phi}_k^H + \mathbf{w}_{pm}^H,
%\vspace*{-0.2cm}
\end{equation}
where $p$ is the transmit power allocated to each pilot symbol and $\mathbf{w}_{pm} \sim \mathcal{CN}(\mathbf{0}, \sigma_a^2 \mathbf{I}_{\tau_p})$ is the additive white Gaussian noise (AWGN) at AP~$m$ having zero mean and standard derivation of $\sigma_a$~[dB]. Furthermore, the training signal received at the GPU from the space link of Fig.~\ref{FigSysModel} used for estimating the satellite UL channel is formulated as
%\vspace*{-0.1cm}
\begin{equation} \label{eq:Yp}
\mathbf{Y}_p = \sum\nolimits_{k=1}^K \sqrt{pK} \mathbf{g}_{k} \pmb{\phi}_k^H + \mathbf{W}_{p},
%\vspace*{-0.1cm}
\end{equation}
where $\mathbf{W}_{p} \in \mathbb{C}^{N \times K}$ models the AWGN, the imperfect feeder link, and the imperfect synchronization between the satellite and terrestrial links with each element distributed as $\mathcal{CN}(0, \sigma_s^2)$. The desired UL channels are estimated both at the APs and the satellite gateway by relying on the minimum mean square error (MMSE) estimation, as shown in Lemma~\ref{Lemma:Est}.\footnote{We consider the MMSE estimation in this framework since it is Baysian estimator, which  minimizes the mean square error (MSE). The performance of the suboptimal channel estimators with lower computational complexity is of interest and left for a future work.} 
%\vspace*{-0.2cm}
\begin{lemma} \label{Lemma:Est}
The MMSE estimate $\hat{g}_{mk}$ of the UL channel $g_{mk}$ between user~$k$ and AP~$m$ can be computed from \eqref{eq:ypm} as
%\vspace*{-0.1cm}
\begin{equation} \label{eq:ChanEstgmk}
\hat{g}_{mk} = \mathbb{E}\{ g_{mk} | \mathbf{y}_{pm}^H\pmb{\phi}_k \} = \frac{\sqrt{pK} \beta_{mk}\mathbf{y}_{pm}^H\pmb{\phi}_k}{p K  \beta_{mk} + \sigma_a^2},
\vspace*{-0.1cm}
\end{equation}
which is distributed as $\hat{g}_{mk} \sim \mathcal{CN}(0, \gamma_{mk})$ and its variance is
%\vspace*{-0.1cm}
\begin{equation}
\gamma_{mk} = \mathbb{E} \{ |\hat{g}_{mk}|^2\} =   \frac{pK \beta_{mk}^2}{ p K  \beta_{mk} + \sigma_s^2}.
%\vspace*{-0.1cm}
\end{equation}
The channel estimation error $e_{mk} = g_{mk} - \hat{g}_{mk}$ is distributed as $e_{mk} \sim \mathcal{CN}(0, \beta_{mk} - \gamma_{mk} )$. Observe that the channel estimate $\hat{g}_{mk}$ and the channel estimation error are independent.

The MMSE estimate $\hat{\mathbf{g}}_k$ of the channel $\mathbf{g}_k$ spanning from user~$k$ to the satellite can be formulated based on \eqref{eq:Yp} as
%\vspace*{-0.1cm}
\begin{equation} \label{eq:ChannelEstgk}
\hat{\mathbf{g}}_k = \bar{\mathbf{g}}_k +  \sqrt{pK} \mathbf{R}_k \pmb{\Phi}_k \big( \mathbf{Y}_{p} \pmb{\phi}_k -   \sqrt{pK} \bar{\mathbf{g}}_k \big),
%\vspace*{-0.1cm}
\end{equation}
where we have $\pmb{\Phi}_k = \big( p K \mathbf{R}_{k} + \sigma_s^2 \mathbf{I}_N \big)^{-1}$. Additionally, the channel estimation error $\mathbf{e}_k = \mathbf{g}_k - \hat{\mathbf{g}}_k$ and the channel estimate $\hat{\mathbf{g}}_k$ are independent random variables, which are distributed as
%\vspace*{-0.1cm}
\begin{align} \label{eq:EstSat}
& \hat{\mathbf{g}}_k \sim \mathcal{CN}(\bar{\mathbf{g}}_k,  p K \mathbf{R}_k \pmb{\Phi}_k \mathbf{R}_k), \mathbf{e}_k \sim \mathcal{CN}(\mathbf{0}, \mathbf{R}_k - p K \pmb{\Theta}_k),
%\vspace*{-0.1cm}
\end{align}
with $\pmb{\Theta}_k = \mathbf{R}_k \pmb{\Phi}_k \mathbf{R}_k, \forall k$.
%\vspace*{-0.2cm}
\end{lemma}
\begin{proof}
The proof follows from adopting the standard MMSE estimation of \cite{Kay1993a} for our system and channel model.
%\vspace*{-0.2cm}
\end{proof}
Given the independence of the channel estimates and estimation errors, this may be  conveniently exploited in our ergodic data throughput analysis and optimization in the next sections. The LoS components of the space links can be estimated very accurately at the satellite gateway from its received training signals. The closed-form expression of the channel estimates obtained in Lemma~\ref{Lemma:Est} will be  utilized to formulate the receivers' combining weights required for detecting the desired signals in the UL satellite and APs. 
%\vspace{-0.4cm} 
\section{Uplink Data Transmission and Ergodic Throughput Analysis} \label{Sec:ULData}
%\vspace{-0.2cm}
This section presents the UL data transmission, where all users send their signals both to the APs and to the satellite in a multiple-access protocol. The UL throughput of each user is derived first for arbitrary signal detection techniques. Then a closed-form expression is obtained for an MRC receiver, which is computationally simple and can be readily implemented in a distributed manner relying on the channel estimation procedure detailed in Section~\ref{Sec:ULP}.\footnote{Other linear combining techniques, such as partial MMSE, offer better performance than MRC, therefore, gaining much research interest. By relying on the random matrix theory, the closed-form expression of the ergodic data throughput of the partial MMSE combining technique can be derived in the asymptotic regime, which may be parts of our future work.} 
%\vspace{-0.4cm}
\subsection{Uplink Data Transmission} \label{Sec:UDT}
%\vspace{-0.2cm}
All the $K$ users transmit their data both to the $M$ APs and to the satellite, where the symbol $s_k$  of user~$k$ obeys  $\mathbb{E}\{ |s_k|^2\} =1$. This data symbol is allocated a transmit power level $\rho_k >0$. The  signal received at the CPU by the space links, denoted by $\mathbf{y} \in \mathbb{C}^N$, and AP~$m$, denoted by $y_m \in \mathbb{C}$, are, respectively, formulated as
%\vspace*{-0.1cm}
\begin{equation} 
\mathbf{y} = \sum\nolimits_{k=1}^K \sqrt{\rho_k} \mathbf{g}_k s_k  + \mathbf{w} \mbox{ and } y_m =  \sum\nolimits_{k=1}^K \sqrt{\rho_k} g_{mk} s_k  + w_m, \label{eq:ReceiveSigAP}
%\vspace*{-0.1cm}
\end{equation}
where $\mathbf{w} \sim \mathcal{CN}(\mathbf{0}, \sigma_s^2 \mathbf{I}_N)$ and $w_m \sim \mathcal{CN}(0,\sigma_a^2)$  represent the AWGN noise and the other imperfections at the satellite receiver and at AP~$m$, respectively. By exploiting \eqref{eq:ReceiveSigAP}, the system  detects the  signal received from user~$k$, $\forall k$ at the CPU  from the following expression
%\vspace*{-0.1cm}
\begin{equation} \label{eq:2ndsk}
\hat{s}_k  = \mathbf{u}_k^H \mathbf{y} + u_{mk} y_m,
%\vspace*{-0.1cm}
\end{equation}
where $\mathbf{u}_k \in \mathbb{C}^{N}$ is the linear detection vector used for inferring the desired signal arriving from the satellite and $u_{mk} \in \mathbb{C}$ is the detection coefficient used by AP~$m$ (see Fig.~\ref{FigSysModel}). The received symbol estimate in \eqref{eq:2ndsk} combines all the different propagation paths, which explicitly unveils the potential benefits of integrating a satellite into terrestrial networks for improving the reliability and/or the throughput. Upon considering only one of the right-hand side terms of \eqref{eq:2ndsk}, the received signal becomes that of a conventional satellite network \cite{abdelsadek2021future} or a terrestrial cooperative network \cite{ngo2017cell}. Thus, we are considering an  advanced cooperative wireless network relying on the coexistence of both space and terrestrial links.
%\vspace{-0.4cm}
\subsection{Uplink Data Throughput}
%\vspace{-0.2cm}
We emphasize that if the number of APs and antennas at the satellite is sufficiently high to treat the channel gain of the desired signal in \eqref{eq:2ndsk} as a deterministic value, the  throughput of user~$k$ can be analyzed conveniently. In order to carry out the throughput analysis, let us first introduce the new variable
%\vspace*{-0.1cm}
\begin{equation} \label{eq:zkkprime}
z_{kk'} = \mathbf{u}_k^H \mathbf{g}_{k'}  +  \sum\nolimits_{m=1}^M  u_{mk}^\ast  g_{mk'},
%\vspace*{-0.1cm}
\end{equation}
which we term as the overall channel coefficient after the use of signal detection techniques, including both the satellite and terrestrial effects. The overall channel coefficient in \eqref{eq:zkkprime} leads to a coherent received signal combination at the CPU. We can assume perfectly phase-coherent symbol-synchronization at all the receivers and neglect any phase-jitter.\footnote{The satellite and terrestrial links have different time delays.  Clearly, given its wide bandwidth, fiber would be faster, requiring some buffering at the CPU before these can be processed when the signal from the satellite reaches the CPU.  The longer propagation delay of the concatenated user-satellite-ground-station path has to be compensated by appropriately delaying the terrestrial signal for coherent combination at the CPU \cite{you2020massive}. The CPU shall be, in principle, able to detect the desired packets from the satellite system by reading headers at higher layers. The time-delay will change due to satellite movement. The synchronization in higher layers can be used to predict the time delay for coherent processing of the next incoming packets, which has to be decoded before reading the headers. The average delay experienced by the satellite is known as a priori (based on geometry), which allows for a coarse synchronization (up to the symbol time period) but possibly fine synchronization might be still needed \cite{choi2015challenges}.  For a LEO satellite at 600~km altitude this only imposes $4$~ms turn-around delay. Moreover, the imperfect symbol-synchronization and the phase impairment  in practice can be compensated by an advanced carrier aggregation technique \cite{kibria2020carrier}.} Otherwise, the imperfect phase-coherent symbol-synchronization can be synthesized by the Gaussian distribution (see Fig.~\ref{Fig:DiffNoiseFloor}). In particular, the desired signal in \eqref{eq:2ndsk} becomes
\begin{multline} \label{eq:Decodesig}
\hat{s}_k  = \sqrt{\rho_k} \mathbb{E} \{ z_{kk} \} s_k +  \sqrt{\rho_k} \left( z_{kk}  -  \mathbb{E} \{ z_{kk} \}  \right)s_k +  \\  \sum\limits_{k'=1, k'\neq k}^K \sqrt{\rho_{k'}} z_{kk'} s_{k'} + \alpha_k^\ast \mathbf{u}_k^H \mathbf{w} + \sum\nolimits_{m=1}^M u_{mk}^\ast w_m,
\end{multline}
where the first additive term contains the desired signal  associated with a deterministic effective channel gain. The second term represents the  beamforming uncertainty, demonstrating the randomness of the effective channel gain for a given signal detection technique. The remaining terms are the mutual interference and noise. By virtue of the  use-and-then-forget channel capacity bounding technique of \cite{Chien2017a,massivemimobook}, the  ergodic  throughput of user~$k$ is
%\vspace*{-0.1cm}
\begin{equation} \label{eq:Rkv1}
	R_k = \left( 1 - K/\tau_c \right) B \log_2 ( 1 + \mathrm{SINR}_k ), \mbox{[Mbps]},
%\vspace*{-0.1cm}
\end{equation}
where $B$ [MHz] is the system bandwidth. The effective signal-to-interference-and-noise ratio (SINR) expression, denoted by $\mathrm{SINR}_k$, is given as
\begin{figure*}
%\vspace*{-0.1cm}
\begin{equation} \label{eq:SINRk}
\mathrm{SINR}_k = \frac{\rho_k | \mathbb{E}\{ z_{kk}\}|^2 }{\sum\nolimits_{k'=1}^K \rho_{k'}  \mathbb{E}\{ |z_{kk'}|^2 \} - \rho_k \big| \mathbb{E}\{ z_{kk}\} \big|^2 +  \mathbb{E} \big\{ \big| \mathbf{u}_k^H \mathbf{w} \big|^2 \big\} + \sum\nolimits_{m=1}^M  \mathbb{E} \big\{ | u_{mk}^\ast w_m |^2 \big\}  }
%\vspace*{-0.1cm}
\end{equation}
%\hrule
\vspace*{-0.7cm}
\end{figure*}
The throughput in \eqref{eq:Rkv1} can be achieved by  arbitrary signal detection techniques at the satellite and APs, since it represents a lower bound of the channel capacity. One can numerically evaluate \eqref{eq:Rkv1} with the aid of  the SINR expression in \eqref{eq:SINRk}, but it requires many realizations of the small-scale fading coefficients to compute several expectations. The direct evaluation of \eqref{eq:Rkv1} relying on Monte Carlo simulations does not provide analytical insights about the impact of the individual parameters on the system performance.
%\vspace{-0.4cm}
\subsection{Uplink Throughput for Maximum Ratio Combining}
%\vspace{-0.2cm}
For gaining further insights, we derive a  closed-form expression for \eqref{eq:Rkv1} by relying on statistical signal processing, when the MRC receiver is used by both the satellite and the AP, i.e., $u_{mk} = \hat{g}_{mk}, \forall m,k,$ and $\mathbf{u}_k = \hat{\mathbf{g}}_k, \forall k,$ as in Theorem~\ref{Theorem:ClosedForm}.
\begin{theorem} \label{Theorem:ClosedForm}
If the MRC receiver is utilized for detecting the desired signal, the UL throughput of user~$k$ is evaluated by  \eqref{eq:Rkv1} with the aid of the effective SINR value obtained in closed form for the given channel statistics as
%\vspace*{-0.1cm}
\begin{equation} \label{eq:ClosedSINR}
 \mathrm{SINR}_k = \frac{\rho_k \left(\|\bar{\mathbf{g}}_k\|^2 +  p K \mathrm{tr}(\pmb{\Theta}_k)  +  \sum_{m=1}^M \gamma_{mk} \right)^2}{ \mathsf{MI}_k + \mathsf{NO}_k},
%\vspace*{-0.1cm}
\end{equation}
where the mutual interference $\mathsf{MI}_k$, and noise $\mathsf{NO}_k$ are respectively given as follows
%\vspace*{-0.1cm}
\begin{align}
\mathsf{MI}_k =&  \sum\nolimits_{k' =1 , k' \neq k }^K \rho_{k'} |\bar{\mathbf{g}}_{k}^H  \bar{\mathbf{g}}_{k'} |^2  + p K \sum\nolimits_{k' =1}^K \rho_{k'}    \bar{\mathbf{g}}_{k'}^H\pmb{\Theta}_k \bar{\mathbf{g}}_{k'} \notag \\
&+ \sum\nolimits_{k' =1}^K \rho_{k'} \bar{\mathbf{g}}_{k}^H \mathbf{R}_{k'} \bar{\mathbf{g}}_{k}  + p K \sum\nolimits_{k' =1 }^K \rho_{k'}   \mathrm{tr}( \mathbf{R}_{k'} \pmb{\Theta}_k ) \notag \\
& +  \sum\nolimits_{k' =1}^K \sum\nolimits_{m=1}^M \rho_{k'}  \gamma_{mk} \beta_{mk'}, \label{eq:MIk}\\
\mathsf{NO}_k =&  \sigma_s^2 \|\bar{\mathbf{g}}_k\|^2 +  p K \sigma_s^2 \mathrm{tr}( \pmb{\Theta}_k ) + \sigma_a^2 \sum\nolimits_{m=1}^M  \gamma_{mk}. \label{eq:NOk}
%\vspace*{-0.1cm}
\end{align}
%\vspace*{-0.2cm}
\end{theorem}
\begin{proof}
The proof is accomplished by computing the expectations in \eqref{eq:SINRk} using the channel models in Section~\ref{Sec:ChannelModel} and the statistical information in Lemma~\ref{Lemma:Est}. The detailed proof is available in Appendix~\ref{Appendix:ClosedForm}.
%\vspace*{-0.2cm}
\end{proof}
The UL throughput of user~$k$ obtained in Theorem~\ref{Theorem:ClosedForm} is a function of the channel statistics, which has a complex expression due to the presence of space links, and it is independent of the small-scale fading coefficients. The spatial correlation and the LoS components created by the presence of the satellite beneficially boost the desired signals, as shown in the numerator of \eqref{eq:ClosedSINR}. The denominator of \eqref{eq:ClosedSINR} represents the interference and noise that degrades the performance, where the SINR is linearly proportional both to  the number of satellite antennas and to the number of APs. Therefore, the achievable throughput of each user should be improved by installing more antennas at the satellite and more APs on the ground. This demonstrates the benefits of distributed APs as shown by the summation of $M$ terms associated with the spatial diversity gain. In the absence of the satellite, the overall channel coefficient is simplified to $z_{kk'} =   \sum\nolimits_{m=1}^M  u_{mk}^\ast  g_{mk'}$, and therefore the effective SINR expression reduces to
%\vspace*{-0.1cm}
\begin{equation} \label{eq:SINRTerrest}
\mathrm{SINR}_k = \frac{\rho_k |\sum_{m=1}^M  \gamma_{mk} |^2}{\sum_{k' =1}^K \sum_{m=1}^M \rho_{k'} \gamma_{mk} \beta_{mk'} + \sigma_a^2 \sum_{m=1}^M  \gamma_{mk}},
%\vspace*{-0.1cm}
\end{equation}
which unveils that the desired signal gain in the numerator  accrues from the centralized signal processing and cooperation among the APs in support of user~$k$. The mutual interference and noise are expressed compactly in the denominator. Hence again, the throughput can be improved by increasing the number of APs. By contrast, in the absence of the APs, the overall channel coefficient is simplified to $z_{kk'} =  \mathbf{u}_k^H \mathbf{g}_{k'}$, and therefore the effective SINR can be expressed as in \eqref{eq:SINRSatelliteOnly}. 
\begin{figure*}
%\vspace*{-0.1cm}
\begin{equation} \label{eq:SINRSatelliteOnly}
\fontsize{10}{10}{\mathrm{SINR}_k = \frac{\rho_k \left|\|\bar{\mathbf{g}}_k\|^2 +  p K \mathrm{tr}(\pmb{\Theta}_k)\right|^2}{ \sum\limits_{k' =1 , k' \neq k }^K \rho_{k'} |\bar{\mathbf{g}}_{k}^H  \bar{\mathbf{g}}_{k'} |^2  + p K \sum\limits_{k' =1}^K \rho_{k'}    \bar{\mathbf{g}}_{k'}^H\pmb{\Theta}_k \bar{\mathbf{g}}_{k'} + \sum\limits_{k' =1}^K \rho_{k'} \bar{\mathbf{g}}_{k}^H \mathbf{R}_{k'} \bar{\mathbf{g}}_{k}  + p K \sum\limits_{k' =1 }^K \rho_{k'}   \mathrm{tr}( \mathbf{R}_{k'} \pmb{\Theta}_k )  +  \sigma_s^2 \|\bar{\mathbf{g}}_k\|^2 +  p K \sigma_s^2 \mathrm{tr}( \pmb{\Theta}_k )  }}
%\vspace*{-0.1cm}
\end{equation}
\hrule
\vspace*{-0.3cm}
\end{figure*}
The desired signal strength is enhanced by both the LoS and NLoS satellite channels, which explicitly shows the benefits of the satellite. 
\vspace*{-0.2cm}
\begin{remark} 
 The coexistence of the satellite and APs generalizes the data throughput analysis of previous works on either  space or terrestrial communications and combines the advantages of both the transmission modes. Coherent data processing at the CPU yields a quadratic array gain on the order of $(M + 2N)^2$. The closed-form expression of the ergodic data throughput in \eqref{eq:ClosedSINR} quantifies the  improvements offered by space-terrestrial communications. For the sake of completeness, we have shown that the stand-alone terrestrial communications only provides an array gain scaling increased with the number of APs, i.e., say $M^2$, while the dominant LoS path in each space link  boosts the array gains with the order of $4N^2$. 
 
Since we consider single-antenna APs, the framework considered leaves room for different deployment conditions and for new resource allocation problems under the space-terrestrial cooperative framework using different beamforming techniques and multiple antennas at the APs. Table~\ref{Table:CompareUL} analytically compares the SINR values for the systems considered.\footnote{
In this paper, we can obtain the exact closed-form solution on the uplink ergodic rate for MRC detection for an arbitrary set of $M,N$, and $K$. We believe that a framework to approximately derive an closed-form expression on the uplink ergodic rate of both the ZF and MMSE detection may indeed be constructed. However, the methodology  would be different since it requires the assumption that $(M+N)/K \rightarrow \infty$ at a fixed rate \cite{Hoydis2013a}. The  closed-form expression on the uplink ergodic rate matches very well with Monte-Carlo simulations for $(M+N)/K \rightarrow \infty$ at a given rate. By contrast, we can obtain the exact closed-form solution of the uplink ergodic rate for the MRC detection for an arbitrary set of $M,N,$ and $K$.  Since the approaches suitable for ZF and MMSE detection are different from that of MRC detection, we would like to leave this exciting issue for our future work.
}  
%\vspace*{-0.2cm}
\end{remark} 
\begin{table*}[t]
	\caption{Comparison of the SINR value among the three systems: Terrestrial communications, satellite communications, and satellite-terrestrial communications} \label{Table:CompareUL}
	\centering
	%\resizebox{\columnwidth}{!}{
	\resizebox{\textwidth}{!}{\begin{tabular}{|c|c|c|c|}
		\hline
	SINR  & Terrestrial Communications    &  Satellite Communications & Satellite-Terrestrial Communications \\ 
		\hline
		Signal & $ \rho_k \left( \sum\limits_{m=1}^M \gamma_{mk} \right)^2$ &  $ \rho_k \left(\|\bar{\mathbf{g}}_k\|^2 +  p K \mathrm{tr}(\pmb{\Theta}_k) \right)^2$ &  $ \rho_k \left(\|\bar{\mathbf{g}}_k\|^2 +  p K \mathrm{tr}(\pmb{\Theta}_k)  +  \sum\limits_{m=1}^M \gamma_{mk} \right)^2$      \\  
		\hline
	Interference & $\sum\limits_{k' =1}^K \sum\limits_{m=1}^M \rho_{k'}  \gamma_{mk} \beta_{mk'}$ & \makecell{$\sum\limits_{k' =1 ,  k' \neq k }^K \rho_{k'} |\bar{\mathbf{g}}_{k}^H  \bar{\mathbf{g}}_{k'} |^2+ p K \sum\limits_{k' =1}^K \rho_{k'}    \bar{\mathbf{g}}_{k'}^H\pmb{\Theta}_k \bar{\mathbf{g}}_{k'} +$\\ $ \sum\limits_{k' =1}^K \rho_{k'} \bar{\mathbf{g}}_{k}^H \mathbf{R}_{k'} \bar{\mathbf{g}}_{k}  + p K \sum\limits_{k' =1 }^K \rho_{k'}   \mathrm{tr}( \mathbf{R}_{k'} \pmb{\Theta}_k )$}  &  \makecell{$\sum\limits_{k' =1 , k' \neq k }^K \rho_{k'} |\bar{\mathbf{g}}_{k}^H  \bar{\mathbf{g}}_{k'} |^2+ p K \sum\limits_{k' =1}^K \rho_{k'}    \bar{\mathbf{g}}_{k'}^H\pmb{\Theta}_k \bar{\mathbf{g}}_{k'} +$\\ $ \sum\limits_{k' =1}^K \rho_{k'} \bar{\mathbf{g}}_{k}^H \mathbf{R}_{k'} \bar{\mathbf{g}}_{k}  + p K \sum\limits_{k' =1 }^K \rho_{k'}   \mathrm{tr}( \mathbf{R}_{k'} \pmb{\Theta}_k ) +$ \\ $  \sum\limits_{k' =1}^K \sum\limits_{m=1}^M \rho_{k'}  \gamma_{mk} \beta_{mk'}$} \\
		\hline
		Noise & $ \sigma_a^2 \sum\limits_{m=1}^M  \gamma_{mk}$ & $ \sigma_s^2 \|\bar{\mathbf{g}}_k\|^2 +  p K \sigma_s^2 \mathrm{tr}( \pmb{\Theta}_k )$& $ \sigma_s^2 \|\bar{\mathbf{g}}_k\|^2 +  p K \sigma_s^2 \mathrm{tr}( \pmb{\Theta}_k ) + \sigma_a^2 \sum\limits_{m=1}^M  \gamma_{mk}$ \\
		\hline
	\end{tabular}}
	%\vspace*{-0.6cm}
\end{table*}
%\vspace*{-0.4cm}
\section{Uplink Data Power Allocation for Space-Terrestrial Communications} \label{Sec:PowerAllocation}
%\vspace*{-0.2cm}
This section considers a pair of optimization problems comprising the max-min fairness  and total transmit power minimization. These optimization problems underline the considerable benefits of a collaboration between the space and terrestrial links under a finite transmit power at each user.
%\vspace*{-0.2cm}
\subsection{Max-Min Fairness Optimization}
%\vspace*{-0.2cm}
Fairness is of paramount importance for planning the networks to provide an adequate throughput for all users by maximizing the lowest achievable ergodic rate. The max-min fairness optimization, which we would like to solve, is formulated as\footnote{In line with authoritative papers in the literature \cite{ngo2017cell, 9136914}, in this paper, we assume that the backhaul is ideal in the sense that it is capable of traffic to carry infinite rate in an error-free manner. Considering a limited backhaul capacity is a potential extension for future work.}
%\vspace*{-0.1cm}
\begin{subequations} \label{Problem:MaxMinQoS}
	\begin{alignat}{2}
		& \underset{ \{ \rho_{k} \} }{\textrm{maximize}} \; \underset{k}{\textrm{min}}
		& & \, \, R_k \label{eq:Obj1} \\
		& \textrm{subject to}
		& & 0 \leq \rho_{k} \leq P_{\mathrm{max},k} \;, \forall k  ,
	\end{alignat}
%\vspace*{-0.1cm}
\end{subequations}
where $P_{\mathrm{max},k}$ is the maximum power that user~$k$ can allocate to each data symbol. Due to the universality of the data throughput expression of \eqref{eq:Rkv1}, Problem~\eqref{Problem:MaxMinQoS} is applicable to any linear receiver combining method. This paper focuses on the MRC method, since we have derived the ergodic throughput with the closed-form SINR expression for each user  shown in \eqref{eq:ClosedSINR}.\footnote{An extension to the other linear combining technique can be accomplished by using the same methodology, but may require extra cost to evaluate the expectations in \eqref{eq:SINRk} numerically.} Based on the upper-level set, the main characteristics of Problem~\eqref{Problem:MaxMinQoS} are given in Lemma~\ref{Lemma:QuasiConcave}. 
%\vspace*{-0.2cm}
\begin{lemma} \label{Lemma:QuasiConcave}
 Problem~\eqref{Problem:MaxMinQoS} is  quasi-concave as the objective function is constructed based on the ergodic UL throughput in \eqref{eq:Rkv1} with the SINR expression in \eqref{eq:ClosedSINR}.
 %\vspace*{-0.2cm}
\end{lemma}
\begin{proof}
The proof is based on the definition of the upper-level set for a quasi-concave problem. The detailed proof is available in Appendix~\ref{Appendix:QuasiConcave}.
%\vspace*{-0.1cm}
\end{proof}
Lemma~\ref{Lemma:QuasiConcave} unveils that the globally optimal solution to Problem~\eqref{Problem:MaxMinQoS} exists and  can be found in polynomial time. We exploit the quasi-concavity to find the most energy-efficient solution. Upon exploiting that $\mathrm{SINR}_k = 2^{\tau_c R_k /(B(\tau_c -K))}, \forall k$, Problem~\eqref{Problem:MaxMinQoS} is reformulated in an equivalent form by exploiting the epigraph representation of \cite[page 134]{Boyd2004a} as follows
%\vspace*{-0.1cm}
\begin{equation} \label{Problem:Epigraphform}
	\begin{aligned}
		& \underset{ \{ \rho_{k} \} }{\textrm{maximize}} 
		& &  \xi \\
		& \textrm{subject to} && \mathrm{SINR}_k \geq \xi , \forall k, \\
		& & & 0 \leq \rho_{k} \leq P_{\mathrm{max},k} \;, \forall k.
	\end{aligned}
%\vspace*{-0.1cm}
\end{equation}
Observe that Problem~\eqref{Problem:MaxMinQoS} handles the minimum data throughput among all the  $K$ users on a logarithmic price scale, whereas \eqref{Problem:Epigraphform} maximizes  the lowest SINR value in a linear scale. Even though Problem~\eqref{Problem:Epigraphform} could be viewed as a geometric program to attain the maximal fairness level, this would impose high computational complexity, since a hidden convex structure should be deployed \cite{van2018joint}. Observe that, for a given value of $\xi = \xi_o$ in the feasible domain, the minimum total transmit power consumption is obtained by the solution of the following optimization problem\footnote{For a feasible value $\xi$ of Problem~\eqref{Problem:Epigraphform}, the total transmit power can be minimized as a consequence of \cite[Lemma~1]{Yates1995a} since the standard interference functions, which are used for updating the power coefficients, are non-increasing with the number of iterations. Consequently, we can leverage this observation to formulate and solve the total transmit power minimization problem in \eqref{Problem:TotalTransmitPower} for a given value $\xi_o$, which is a solution to Problem~\eqref{Problem:Epigraphform}.}
%\vspace*{-0.1cm}
\begin{subequations}\label{Problem:TotalTransmitPower}
	\begin{alignat}{2}
		& \underset{ \{ \rho_{k} \} }{\textrm{minimize}} 
		& &  \sum\nolimits_{k=1}^K \rho_k \\
		& \textrm{subject to} && \, \, \mathrm{SINR}_k \geq \xi_o , \forall k, \label{eq:SINRConstraints}\\
		& & & 0 \leq \rho_{k} \leq P_{\mathrm{max},k} \;, \forall k.
	\end{alignat}
%\vspace*{-0.1cm}
\end{subequations}
After that, the most energy-efficient solution of Problem~\eqref{Problem:Epigraphform} should be obtained by finding the maximum value of the variable $\xi$ of using, for example, the popular bisection method. The objective function of  \eqref{Problem:TotalTransmitPower} is a linear combination of all the data power variables $\{\rho_k\}, \forall k$. The power constraints are affine, and the SINR constraint of each user can be reformulated as
%\vspace*{-0.1cm}
\begin{equation}
\xi_o \mathsf{MI}_k  + \xi_o \mathsf{NO}_k \leq \rho_k \left(\|\bar{\mathbf{g}}_k\|^2 +  p K \mathrm{tr}(\pmb{\Theta}_k)  +  \sum\nolimits_{m=1}^M \gamma_{mk} \right)^2,
%\vspace*{-0.1cm}
\end{equation}
which is also affine. Consequently, \eqref{Problem:TotalTransmitPower} is  a linear program. A canonical algorithm can get the global solution by the classic interior-point method. The main cost in each iteration is associated with computing the first derivative of the SINR constraints \eqref{eq:SINRConstraints}, which might still impose high computational complexity. Subsequently, in this paper, we propose a low complexity algorithm based on the alternating optimization approach and the closed-form solution for each power coefficient by virtue of the standard interference function (see Definition~\ref{Def:SIF}). By stacking all the transmit data powers in a vector $\pmb{\rho} = [\rho_1, \ldots, \rho_K] \in \mathbb{R}_+^K$, the SINR constraint of user~$k$ is reformulated as
\vspace*{-0.1cm}
\begin{equation} \label{eq:SIFPrice}
\rho_k \geq I_k (\pmb{\rho}),
\vspace*{-0.1cm}
\end{equation}
where $I_k (\pmb{\rho})$ is the standard interference function defined for user~$k$, which is given by
\vspace*{-0.1cm}
\begin{equation}\label{eq:Ikrho}
I_k (\pmb{\rho}) = \frac{\xi_o\mathsf{MI}_k (\pmb{\rho}) + \xi_o \mathsf{NO}_k)}{\left|\|\bar{\mathbf{g}}_k\|^2 +  p K \mathrm{tr}(\pmb{\Theta}_k)  +  \sum\nolimits_{m=1}^M \gamma_{mk} \right|^2},
\vspace*{-0.1cm}
\end{equation}
where the detailed expression of $\mathsf{MI}_k (\pmb{\rho})$ has been given in \eqref{eq:MIk}, but here we express it as a function of the transmit power variables stacked in $\pmb{\rho}$. Apart from the SINR constraint, the data power of each user should satisfy the individual power budget, hence we have
\vspace*{-0.1cm}
\begin{equation} \label{eq:PowerConstraint}
 I_k (\pmb{\rho}) \leq \rho_k \leq P_{\max,k}.
\vspace*{-0.1cm}
\end{equation}
One can search across the range of each data power variable observed in  \eqref{eq:PowerConstraint}, where the global optimum  of Problem~\eqref{Problem:Epigraphform} is validated by Theorem~\ref{Theorem:Bisection}.
\vspace*{-0.2cm}
\begin{theorem} \label{Theorem:Bisection}
For a given feasible $\xi_o$ value and the initial data powers $\rho_k(0) = P_{\max,k}, \forall k$, the globally optimal solution of Problem~\eqref{Problem:TotalTransmitPower} is obtained by computing the standard interference function in \eqref{eq:Ikrho} and the power constraint in \eqref{eq:PowerConstraint} for all users. In more detail, if the data power of user~$k$ is updated at iteration~$n$ as   
\vspace*{-0.1cm}
\begin{equation} \label{eq:rho}
\rho_k(n) = I_k(\pmb{\rho}(n-1)),
\vspace*{-0.1cm}
\end{equation}
where $I_k(\pmb{\rho}(n-1))$ is defined in \eqref{eq:Ikrho} with $\pmb{\rho}(n-1)$ denoting the data power vector from the previous iteration, then this iterative approach converges to the unique optimal solution after a finite number of iterations.  Owning to the feasibility of $\xi_o$, it holds that $I_k(\pmb{\rho}(n-1)) \leq P_{\max,k}, \forall k$.
\begin{algorithm}[t]
	\caption{Data power allocation to  Problem~\eqref{Problem:MaxMinQoS} by using the standard interference function and the bisection method} \label{Algorithm1}
	\textbf{Input}:  Define the maximum data powers $P_{\max,k}, \forall k$; Select initial values $\rho_{k}(0) = P_{\max,k}, \forall k$; Set the maximum bound $\xi_{o}^{\mathrm{up}}$ as in \eqref{eq:xioup}; Define $\xi_{\min,o} = 0$ and $\xi_{\max,o} = \xi_{o}^{\mathrm{up}}$; Set the inner tolerance $\epsilon$ and the outer tolerance $\delta$.
	\begin{itemize}
		\item[1.]  Initialize the outer loop index $n=1$.
		 \item[2.] \textbf{while} $\xi_{\max,o} - \xi_{\min,o} > \delta$ \textbf{do}
		\begin{itemize}
  	        \item[2.1.] Set $\tilde{\rho}_k(0) = \rho_k(0), \forall k$; Set $\xi_o = (\xi_{\min,o}  + \xi_{\max,o} )/2$ and compute $R_o =  B(1-K/\tau_c)\log_2(1+\xi_o)$.
			\item[2.2.] Compute the total power consumption $P_{\mathrm{tot}}(0) =\sum_{k=1}^K \rho_{k}(0)$.
			\item[2.3] Initialize the accuracy $T= P_{\mathrm{tot}}(0)$ and the inner loop index $m=1$.
			\item[2.4.] \textbf{while} $T> \varepsilon$ \textbf{do} 
			\begin{itemize}
			\item[2.4.1.] User~$k$ computes the standard interference function $	{I}_{k} \left(\tilde{\pmb{\rho}} (m-1) \right)$ using \eqref{eq:rho} with $\tilde{\pmb{\rho}} (m-1) =[\tilde{\rho}_1 (m-1), \ldots, \tilde{\rho}_K (m-1)] \in \mathbb{R}_+^K$.
			\item[2.4.2.] User~$k$ updates its temporary data power as \eqref{eq:rhotilde}.
			\item[2.4.3.] Repeat Steps $2.4.1$ and $2.4.2$ with other users, then update the accuracy as in \eqref{eq:gamman}.
		     \item[2.4.4.] If $T \leq \varepsilon$ $\rightarrow$ Compute $R_k(\tilde{\pmb{\rho}}_{k}(m)),\forall  k,$ and go to Step~$3$. Otherwise, set $m= m+1$ and go to Step $2.4.1$.
			\end{itemize}
		    \item[2.5.] \textbf{End while}
			\item[2.6.] If $\exists k, R_k(\tilde{\pmb{\rho}}_{k}(m)) < R_o$, set $\xi_{\max,o} = \xi_o$ and go to Step 1. Otherwise, update $\rho_{k}(n) = 	\tilde{\rho}_k(m), \forall k,$ and set $\xi_{\min,o} = \xi_o$, set $n=n+1$, and go to Step~$1$.
			
		\end{itemize}
		\item[3.] \textbf{End while}
		\item[4.]  Set $\rho_{k}^{\ast} = \rho_{k}(n),\forall  k$.
	\end{itemize}
	\textbf{Output}: Final interval $[\xi_{\min,o}, \xi_{\max,o}]$ and $\{ \rho_{k}^{\ast} \}$, $\forall l,k$. \vspace*{-0.0cm}
\end{algorithm}

Then,  the most energy-efficient solution of Problem~\eqref{Problem:Epigraphform} is obtained by updating the lower bound of the SINR values across the search range $\xi_o \in [0,  \xi_o^{\mathrm{up}}]$, where $\xi_o^{\mathrm{up}}$ is given by
\vspace*{-0.1cm}
\begin{equation} \label{eq:xioup}
\xi_o^{\mathrm{up}} =  \underset{k}{\min} \, \frac{P_{\max,k} \left|\|\bar{\mathbf{g}}_k\|^2 +  p K \mathrm{tr}(\pmb{\Theta}_k)  +  \sum_{m=1}^M \gamma_{mk} \right|^2}{ \mathsf{NO}_k}.
\vspace*{-0.1cm}
\end{equation}
Along all considered values of the variable $\xi_0$, Problem~\eqref{Problem:TotalTransmitPower} is infeasible if the following condition is met at least by one user for a given value $\xi_o$ as
\vspace*{-0.1cm}
 \begin{equation} \label{eq:IkPmaxk}
 R_k (\pmb{\rho}(n)) < B(1-K/\tau_c)\log_2(1+\xi_o), \exists k \in \{1, \ldots, K\}.
 \vspace*{-0.1cm}
 \end{equation}
\end{theorem}
\begin{proof}
The main proof hinges on verifying the standard interference function defined for each user and on finding an upper bound for the bisection method. The infeasiblity detection is then straightforwardly obtained. The detailed proof is available in Appendix~\ref{Appendix:Bisection}.
\end{proof}
Theorem~\ref{Theorem:Bisection} provides an iterative design to  obtain the global optimum of Problem~\eqref{Problem:MaxMinQoS}: Firstly, a low complexity mechanism  is presented to update the data powers along the iterations as seen in \eqref{eq:rho}. Secondly, we may use an efficient search, for example, the bisection method for maximizing the minimum data throughput with an  effective search range whose upper bound given in \eqref{eq:xioup}. Thirdly, the achievable rate condition \eqref{eq:IkPmaxk} can be exploited to detect if Problem~\eqref{Problem:TotalTransmitPower} is feasible for a given value $\xi_o$  by updating the data throughput in each iteration.

From the analytical features in Theorem~\ref{Theorem:Bisection}, the proposed alternating technique of finding the optimal solution to Problem~\eqref{Problem:MaxMinQoS} is shown in Algorithm~\ref{Algorithm1} by initially setting the maximum data power to each user, i.e., $\rho_k(0) = P_{\max,k}, \forall k$ and the range of $[\xi_{\min,o}, \xi_{\max,o}]$ for the parameter $\xi_o$ by capitalizing on \eqref{eq:xioup}. The bisection is utilized to update $\xi_o$, while the data powers are iteratively updated by the standard interference function seen in \eqref{eq:rho} subject to the condition \eqref{eq:PowerConstraint}. In particular, for a given value of $\xi_o = (\xi_{\min,o} + \xi_{\max,o})/2$, the temporary data power coefficients  are set as $\tilde{\rho}_k(0) = \rho_k(0), \forall k$. Then user~$k$ will update its temporary data power at inner iteration~$m$ as
\vspace*{-0.1cm}
\begin{equation} \label{eq:rhotilde}
	\tilde{\rho}_k(m) = \min({I}_{k} \left(\tilde{\pmb{\rho}} (m-1) \right), P_{\max,k}).
\vspace*{-0.1cm}
\end{equation}
The inner loop can be terminated, when the difference between two consecutive iterations becomes small. For example, we may compute the normalized total power consumption ratio
%\vspace*{-0.1cm}
\begin{equation} \label{eq:gamman}
\begin{split}
T &= \frac{| P_{\mathrm{tot}}(m) - P_{\mathrm{tot}}(m-1) |}{  P_{\mathrm{tot}}(m-1)}  \\
& = \frac{\left|\sum_{k=1}^K \tilde{\rho}_{k}(m)  - \sum_{k=1}^K \tilde{\rho}_{k}(m-1)  \right|}{\sum_{k=1}^K \tilde{\rho}_{k}(m-1) },
\end{split}
%\vspace*{-0.1cm}
\end{equation}
with $P_{\mathrm{tot}}(m) = \sum_{k=1}^K \tilde{\rho}_{k}(m) $ denoting the total transmit data power at inner iteration~$m$. The data power will converge to the optimal solution with tolerance as $\gamma(n) \leq \epsilon$. For a given value $\xi_o$, if Problem~\eqref{Problem:TotalTransmitPower} is feasible, the lower bound of $\xi_o$ is then updated, yielding $\xi_{\min,o} = \xi_o$ after obtaining the data power solution. Otherwise, \eqref{eq:IkPmaxk} is utilized  to detect if there is no solution to Problem~\eqref{Problem:TotalTransmitPower} for a given $\xi_o$. The closed-form expression in updating the temporary data power as shown in \eqref{eq:rhotilde} will be the data power of user~$k$ if $R_k(\tilde{\pmb{\rho}}_{k}(m)) \geq R_o, \forall k$. It means that outer iteration~$n$ will perform  $\rho_k(n) = \tilde{\rho}_k(m)$ 
and update $\xi_{\min,o} = \xi_o$. Otherwise, the upper bound $\xi_{\max,o}$ will be shrunk as $\xi_{\max,o} =\xi_o$. Assuming that the dominant arithmetic operators are multiplications and divisions, we can estimate the computational complexity of Algorithm~\ref{Algorithm1}, where the channels' statistical information is computed in advance. Specifically, the data power should be acquired at the computational complexity order of $\mathcal{O}((4K^2  + MK^2  + 5K)(U_1+1))$, where $U_1$ is the number of iterations, when Algorithm~\ref{Algorithm1} reaches the accuracy~$\epsilon$. In addition, the bisection method requires a number of iterations that is proportional to $\lceil \log_2 (\xi_o^{\mathrm{up}}/\delta) \rceil $. This demonstrates the effectiveness of using the upper bound \eqref{eq:xioup} in reducing the total cost. Consequently, the computational complexity of Algorithm~\ref{Algorithm1}, denoted by $C_1$, is of the order of
\vspace*{-0.1cm}
\begin{equation} \label{eq:C1}
C_1 = \mathcal{O}\left( \lceil \log_2 (\xi_o^{\mathrm{up}}/\delta) \rceil  \left( (4K^2 + MK^2 + 5K)(U_1 + 1)\right) \right),
\vspace*{-0.1cm}
\end{equation}
where $ \lceil \cdot \rceil $ is the ceiling function. The computational complexity is directly proportional to the number of APs and it is in a quadratic function of the number of users. However, \eqref{eq:C1} emphasizes that the computational complexity of our proposed algorithm does not depend on the number of satellite antennas.
\vspace*{-0.2cm}
\begin{remark}
Algorithm~\ref{Algorithm1} offers a low complexity design that optimizes the max-min fairness service to all the $K$ users  in a space-terrestrial communication system with the most energy-efficient solution based on the analysis of its quasi-concavity. Our proposed design can also detect infeasible problems. The data powers are updated in a closed-form solution by utilizing the standard interference function for a given lower bound on the ergodic throughput. All the users get the uniformly best quality of service by exploiting the bisection method. We emphasize that the max-min fairness optimization always provides a feasible data power allocation solution. However, this optimization problem is not scalable in a sense that for large-scale networks supporting many users, the max-min fairness level tends zero, i.e., $\xi \rightarrow 0$ as $K \rightarrow \infty $, by the use of Theorem~\ref{Theorem:Bisection} under near-far effects.
\vspace*{-0.2cm}
\end{remark}
\begin{algorithm}[t]
	\caption{Data power allocation to problem~\eqref{Problem:TotalV1} by spending maximum transmit power on unsatisfied users} \label{Algorithm2}
	\textbf{Input}:  Define the maximum data powers $P_{\max,k}, \forall k$; Select the initial values $\rho_{k}(0) = P_{\max,k}, \forall k$; Compute the total transmit power consumption $P_{\mathrm{tot}}(0) = \sum_{k=1}^K \rho_{k}(0)$; Set initial value $n=1$ and the tolerance $\epsilon$.
	\begin{itemize}
		\item[1.] User~$k$ computes the standard interference function $	\tilde{I}_{k} \left(\pmb{\rho} (n-1) \right)$ using \eqref{eq:Ikrhov1}.
		\item[2.] If $\tilde{I}_{k} \left(\pmb{\rho} (n-1) \right) > P_{\max,k}$, update $\rho_{k}(n) = P_{\max,k}$. Otherwise, update $ \rho_{k}(n) = 	\tilde{I}_{k} \left(\pmb{\rho} (n-1) \right) $.
		\item[3.] Repeat Steps $1,2$ with other users, then compute the ratio $\gamma (n) =$ $| P_{\mathrm{tot}}(n) - P_{\mathrm{tot}}(n-1) | /  P_{\mathrm{tot}}(n-1)$.
		\item[4.] If $\gamma (n) \leq \epsilon$ $\rightarrow$ Set $\rho_{k}^{\ast} = \rho_{k}(n),\forall k,$ and Stop. Otherwise, set $n= n+1$ and go to Step $1$.
	\end{itemize}
	\textbf{Output}: A fixed point $\rho_{k}^{\ast}$, $\forall k$. \vspace*{-0.0cm}
\end{algorithm}
%\vspace*{-0.4cm}
\subsection{Total Transmit Power Minimization Under Individual Demand-based Constraints}
%\vspace*{-0.3cm}
Another ambition of the future wireless systems is to provide an individual QoS for each user in the coverage area, whilst consuming as little power as possible. A total transmit power minimization problem under the SINR constraints is formulated as
%\vspace*{-0.1cm}
\begin{subequations}\label{Problem:TotalV1}
	\begin{alignat}{2}
		& \underset{ \{ \rho_{k} \} }{\textrm{minimize}} 
		& &  \sum\nolimits_{k=1}^K \rho_k \label{eq:ObjectFunc}\\
		& \textrm{subject to} && \, \, \mathrm{SINR}_k \geq \xi_k , \forall k, \label{eq:SINRConstraintsV1}\\
		& & & 0 \leq \rho_{k} \leq P_{\mathrm{max},k} \;, \forall k,
	\end{alignat}
\end{subequations}
where $\xi_k >0$ denotes the SINR value corresponding to the user-specific data throughput that user~$k$ requests from the network. For a predetermined set of the requested SINR values $\{ \xi_k\}$, Problem~\eqref{Problem:TotalV1} is a linear program having a compact feasible set. Hence, the global optimum always exists and can be obtained in polynomial time, as mentioned. We notice that problem~\eqref{Problem:TotalV1} simultaneously optimizes the data powers of all the users. However, the individual SINR constraints  make it more challenging to guarantee finding the global solution under finite network dimensions. In multiple access scenarios, the system may not be able to serve all the users owing to, for example, the near-far effects, weak channel conditions, and excessive SINR requirements. When the user-specific SINR constraints of some users cannot be satisfied, we arrive at an infeasible solution. Under congestion, the $K$ users are split into: \textit{Satisfied users} who can have their throughput requirements met; and \textit{unsatisfied users} who are served at a throughput less than requested. We conceive a pair of algorithms for detecting congestion and relaxing the SINR requirements of unsatisfied users. The proposed designs can still offer satisfactory SINRs. 
%\vspace*{-0.1cm}
\subsubsection{Assigning Maximum Power to Unsatisfied Users}
Following a similar methodology as for the network-wise SINR constraint in \eqref{eq:Ikrho}, we construct the standard interference function for user~$k$  as
%\vspace*{-0.2cm}
\begin{equation}\label{eq:Ikrhov1}
	\tilde{I}_k (\pmb{\rho}) = \frac{\xi_k\mathsf{MI}_k (\pmb{\rho}) + \xi_k \mathsf{NO}_k}{\left|\|\bar{\mathbf{g}}_k\|^2 +  p K \mathrm{tr}(\pmb{\Theta}_k)  +  \sum\nolimits_{m=1}^M \gamma_{mk} \right|^2}.
%\vspace*{-0.1cm}
\end{equation}
Observe by comparing the standard interference function of \eqref{eq:Ikrhov1} and \eqref{eq:Ikrho} that in \eqref{eq:Ikrhov1} the individual SINR $\xi_k, \forall k,$ are used. A fixed point solution to problem~\eqref{Problem:TotalV1} is obtained in Lemma~\ref{lemma:Alg2}.
%\vspace*{-0.3cm}
\begin{lemma} \label{lemma:Alg2}
From the initial value $\rho_k(0) = P_{\max,k}, \forall k$, if the data power of user~$k$ is updated at iteration $n$ as
\begin{equation} \label{eq:rhok}
\rho_k (n) = \max(\tilde{I}_k(\pmb{\rho}(n-1)), P_{\max,k}),
\end{equation}
then the iterative approach converges to a fixed point solution in polynomial time.
%\vspace*{-0.2cm} 
\end{lemma}
\begin{proof}
As $\tilde{I}_k (\pmb{\rho})$ is a standard interference function, $\hat{I}_k  = \max(\tilde{I}_k(\pmb{\rho}), P_{\max,k})$ is also a standard interference function. From the initial data powers, $\rho_k (0) = P_{\max,k}, \forall k,$  iteration~$n$ updates the data power of user~$k$ as
%\vspace*{-0.1cm}
\begin{equation}
\rho_k (n) = \hat{I}_k( \pmb{\rho}(n-1)).
%\vspace*{-0.1cm}
\end{equation}
If $\hat{I}_k( \pmb{\rho}(n-1)) = P_{\max,k}$, the data power of user~$k$ at iteration~$n$ is $\rho_k (n) = P_{\max,k}$, which still ensures the non-increasing property of the objective function in \eqref{eq:ObjectFunc}. Otherwise, it holds that $\rho_k (n) = \tilde{I}_k(\pmb{\rho} (n-1))$. The convergence is guaranteed by utilizing a similar claim as in Theorem~\ref{Theorem:Bisection}. The proof is complete.
%\vspace*{-0.2cm}
\end{proof}
The analytical result in Lemma~\ref{lemma:Alg2} can be exploited to find a fixed-point solution to Problem~\eqref{Problem:TotalV1} that is implemented in Algorithm~\ref{Algorithm2}. It can work for both feasible and infeasible domains. Once congestion appears, we can detect unsatisfied users by computing the standard interference function and then comparing it to the maximum power allocated to each data symbol. The updated policy \eqref{eq:rhok} indicates that unsatisfied users will be served at a lower throughput than requested. The maximum transmit data power is allocated to unsatisfied users. By counting  the dominant arithmetic operations, namely, the multiplications, the divisions, and the maximum of two numbers, the total computational complexity of Algorithm~\ref{Algorithm2} is on the order of 
%\vspace*{-0.2cm}
%\begin{equation}
$C_2  = \mathcal{O}( MK^2 U_2 + 4K^2 U_2 + 6KU_2 )$,
%\vspace*{-0.1cm}
%\end{equation}
where $U_2$ is the number of iterations required by  Algorithm~\ref{Algorithm2} to reach the accuracy~$\epsilon$, which proves that the cost scales up with the number of APs and in a quadratic order of the number of users.
\begin{algorithm}[t]
	\caption{Data power allocation to problem~\eqref{Problem:TotalV1} by softly removing unsatisfied users} \label{Algorithm3}
	\textbf{Input}:  Define maximum powers $P_{\max,k}, \forall k$; Select initial values $\rho_{k}(0) = P_{\max,k}, \forall k$; Compute the total transmit power consumption $P_{\mathrm{tot}}(0) =  \sum_{k=1}^K \rho_{k}(0)$; Set initial value $n=1$ and tolerance $\epsilon$.
	\begin{itemize}
		\item[1.] User~$k$ computes the standard interference function $	\tilde{I}_{k} \left(\pmb{\rho} (n-1) \right)$ using \eqref{eq:Ikrhov1}.
		\item[2.] If $\tilde{I}_{k} \left(\mathbf{p} (n-1) \right) > P_{\max,k}$, compute $\mu_k$ by using \eqref{eq:muk} and update $\rho_{k}(n) = P_{\max,k}^2 / (\mu_k\tilde{I}_k(\pmb{\rho}(n-1)))$. Otherwise, update $ \rho_{k}(n) = 	\tilde{I}_{k} \left(\pmb{\rho} (n-1) \right) $.
		\item[3.] Repeat Steps $1,2$ with other users, then compute the ratio $\gamma (n) =$ $| P_{\mathrm{tot}}(n) - P_{\mathrm{tot}}(n-1) | /  P_{\mathrm{tot}}(n-1)$.
		\item[4.] If $\gamma (n) \leq \epsilon$ $\rightarrow$ Set $\rho_{k}^{\ast} = \rho_{k}(n),\forall k,$ and Stop. Otherwise, set $n= n+1$ and go to Step $1$.
	\end{itemize}
	\textbf{Output}: A fixed point $\rho_{k}^{\ast}$, $\forall k$. \vspace*{-0.0cm}
\end{algorithm}
\subsubsection{Softly Removing Unsatisfied Users} Each unsatisfied user reduces the data power rather than allocating the maximum power as done in Algorithm~\ref{Algorithm1}. The network can  minimize mutual interference, and therefore ameliorate the number of satisfied users. The idea of softly removing unsatisfied users is analytically characterized in Theorem~\ref{Theorem:SoftRemoval}.
%\vspace*{-0.3cm}
\begin{theorem}\label{Theorem:SoftRemoval}
Commencing from the initial power value of $\rho_k (0) = P_{\max,k}, \forall k$, the data power of user~$k$ is updated at iteration~$n$ as
%\vspace*{-0.1cm}
\begin{equation} \label{eq:rhokn}
\begin{split}
&\rho_k(n) = f_k(\pmb{\rho}(n-1))  \\
&= \begin{cases}
\tilde{I}_k(\pmb{\rho}(n-1)), & \mbox{if } \tilde{I}_k(\pmb{\rho}(n-1)) \leq P_{\max,k},\\
\frac{P_{\max,k}^2}{\mu_k\tilde{I}_k(\pmb{\rho}(n-1))} , & \mbox{if } \tilde{I}_k(\pmb{\rho}(n-1)) \geq P_{\max,k},
\end{cases}
\end{split}
%\vspace*{-0.1cm}
\end{equation}
where $\mu_k \geq 1$ stands for the soft removal rate of unsatisfied user~$k$. The iterative approach converges to a fixed point solution in polynomial time.
\end{theorem}
%\vspace*{-0.2cm}
\begin{proof}
The proof relies on verifying the two-sided scalable function defined for each user and following by the convergence. The detailed proof is available in Appendix~\ref{Appendix:SoftRemoval}.
\end{proof}
At the beginning, all the users will be treated equally and update their data powers by using the standard interference function in \eqref{eq:Ikrhov1}, i.e., $\rho_k(n) =  \tilde{I}_k(\pmb{\rho}(n-1))$. The treatment should be changed if the standard interference function exceeds the limited power. If user~$k$ is found to be unsatisfied in iteration~$n$, its data power will be scaled down as shown in \eqref{eq:rhokn}. The controllable value $\mu_k$ stands for the soft removal rate that is computed  by utilizing, for example, the actual and requested SINR values, which can be defined as
\vspace*{-0.1cm}
\begin{equation} \label{eq:muk}
\mu_k = \xi_k / \mathrm{SINR}_k (\hat{\pmb{\rho}}(n-1)),
\vspace*{-0.1cm}
\end{equation}
where the SINR value of user~$k$ is defined in \eqref{eq:SINRSatelliteOnly} with the data power vector $\hat{\pmb{\rho}}(n-1)$. From \eqref{eq:muk}, the data power of unsatisfied users will be dramatically degraded if the offered  throughput is much lower than their requirements. The total complexity of this algorithm design is dominated by computing the standard interference function and the soft removal rate. Consequently, Algorithm~\ref{Algorithm3}  has the computational complexity order of 
%\vspace*{-0.1cm}
%\begin{equation}
$C_3 = \mathcal{O}(MK^2(Z+1)U_3 + 4K^2(Z+1) U_3 + K(7Z +6)U_3 )$,
%\vspace*{-0.1cm}
%\end{equation}
where $Z$ is the total number of soft removal calculations, and $U_2$ is the number of iterations required by Algorithm~\ref{Algorithm2} to reach the accuracy~$\epsilon$. The two algorithms handling congestion control have a  similar complexity, if the soft removal rate is set to one for all the unsatisfied users. However, if \eqref{eq:muk} is exploited, the latter is more complex than the former, because a progressive policy is applied for scaling down the data powers of unsatisfied users.\footnote{The optimal solution is to demonstrate the fairness of the satellite-terrestrial networks and the minimum power consumption for a given set of individual data throughput demands thereby providing a benchmark by which more practical, near-optimal solutions can be compared.  Moreover, to the best of our knowledge, this is the first work in the satellite-terrestrial literature that has considered the associated congestion issue.  Even though the optimality and the convergence analysis can be applied to networks having an arbitrary number of users, its practical implementation becomes challenging upon increasing the network dimensions.  A potential direction to overcome this challenge is that exploiting a data-driven approach for reducing the computational complexity  by orders of magnitude \cite{van2020power}, which will be explored in our future work.}
%\vspace*{-0.2cm} 
\begin{remark}
This paper considers fast fading space-terrestrial channels where the ergodic data throughput is of particular interest and it is computed by averaging over many different realizations of the small-scale fading coefficients.
The two optimization problems considered allocate the data powers to all users in the network based on the channel statistics that are stable for a long period of time. Similar optimization problems have been considered in the space-terrestrial communications scenarios of \cite{bui2021robust,du2018secure}, but on a short-time scale and the systems communicate over slow fading channels and assuming perfect CSI. The data powers are optimized for a specific instantaneous channel rate, so the solution must be updated to adapt to the envelope changes, whenever the small-scale fading coefficients fluctuate.
%\vspace*{-0.2cm}
\end{remark}
\begin{remark}
%\vspace*{-0.2cm}
The congestion issues over long time scales have not yet been studied in satellite-terrestrial communication systems. Hence, this paper investigates two different metrics defined for quantifying the  satisfaction both of the individual users and of the entire network, obtained by solving problem~\eqref{Problem:TotalV1}. For user-specific demand satisfaction, the satisfied user set $\mathcal{K}_s \subseteq \mathcal{K} = \{1,\ldots K \}$ is defined as
%\vspace*{-0.1cm}
\begin{equation} \label{eq:Ks}
\mathcal{K}_s = \left\{ k \big| R_{k}(\{\rho_k^\ast \}) = B \left(1- K/\tau_c \right) \log_2( 1+ \xi_k), k \in \mathcal{K} \right\},
\vspace*{-0.1cm}
\end{equation}
where $\{\rho_k^{\ast}\}$ is the optimized data power obtained by Algorithm~\ref{Algorithm2} or \ref{Algorithm3}. Furthermore, the unsatisfied user set, denoted by $\mathcal{K}_u \subseteq \mathcal{K}$ is defined as
%\vspace*{-0.1cm}
\begin{equation} \label{eq:Ku}
\mathcal{K}_u =   \left\{ k \big| R_{k}(\{\rho_k^\ast \}) < B\left(1- K/\tau_c \right) \log_2( 1+ \xi_k), k \in \mathcal{K} \right\}.
%\vspace*{-0.1cm}
\end{equation}
For quantifying the entire network's demand satisfaction, Jain's fairness index \cite{bui2021robust, jain1984quantitative} is adopted for our framework:
%\vspace*{-0.1cm}
\begin{equation} \label{eq:JainFairnessIndex}
J = \frac{\left( |\mathcal{K}_s|  + \sum\nolimits_{k \in \tilde{\mathcal{K}}_u} R_k (\{ \rho_k^\ast \}) /\hat{\xi}_k \right)^2}{K |\mathcal{K}_s| +  K \sum_{k \in \tilde{\mathcal{K}}_u}  R_k (\{ \rho_k^\ast \})^2/ \hat{\xi}_k^2}.
%\vspace*{-0.1cm}
\end{equation} 
Jain's fairness index spans from the worst case to the best case in the range [$1/K$, $1$]. 
%\vspace*{-0.2cm}
\end{remark}
\begin{figure}[t]
	\centering
	\includegraphics[trim=2.4cm 8.5cm 7.5cm 8.2cm, clip=true, width=2.4in]{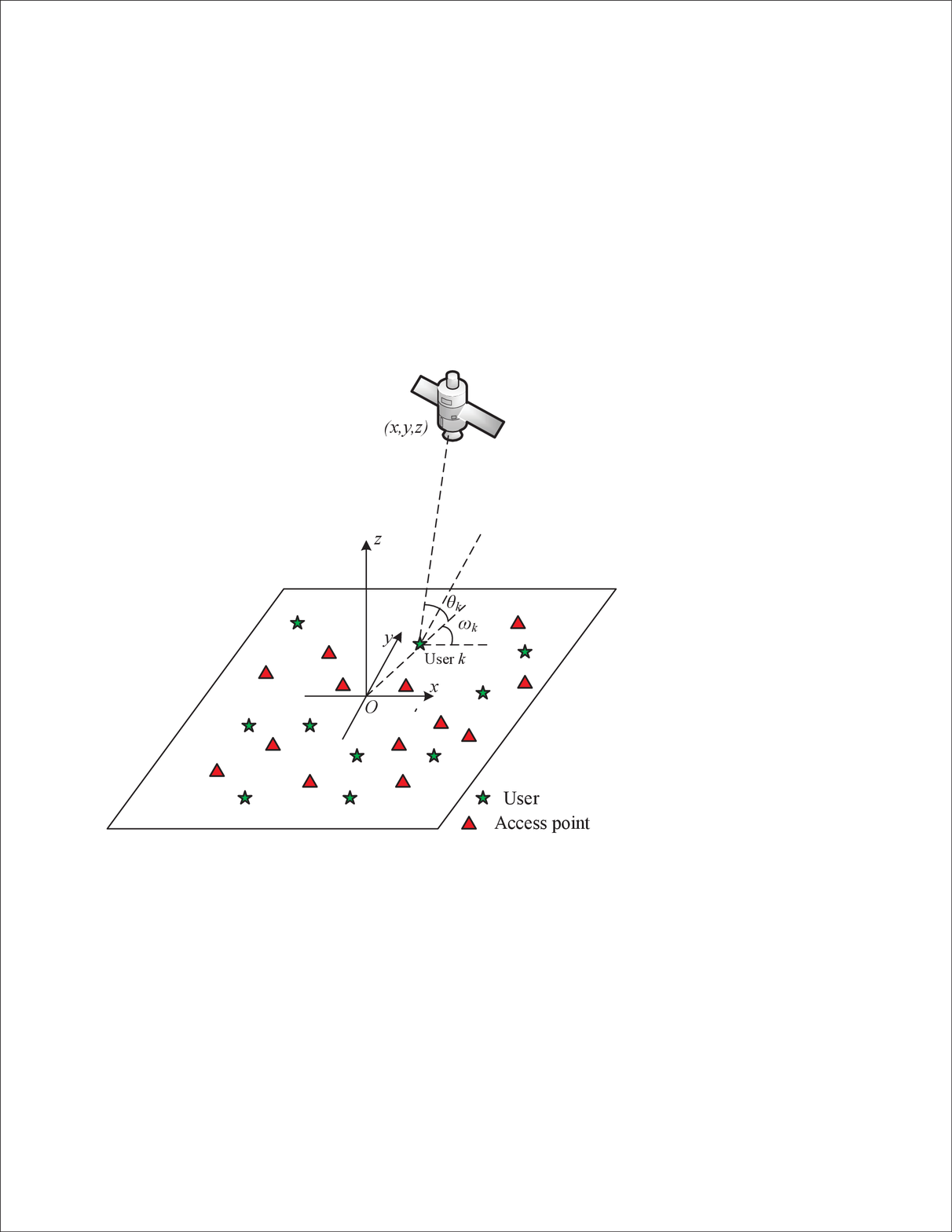} \vspace*{-0.2cm}
	\caption{The satellite-terrestrial  system considered in the simulations, where the locations of satellite, APs, and users are mapped into a three-dimensional (3D) Cartesian coordinate.}
	\label{FigTopology}
	\vspace*{-0.2cm}
\end{figure}
%\vspace*{-0.4cm}
\section{Numerical Results} \label{Sec:NumericalResults}
%\vspace*{-0.2cm}
In order to evaluate the terrestrial links under near-far effects, similar to \cite{ngo2017cell, Chien2021TWC} and reference herein, we study the system performance of a network's deployment in a rural area having $40$ APs and $20$ users uniform randomly distributed in a square area of $20$~km$^2$, mapped into a Cartesian coordinate system $(x,y,z)$ as shown in Fig.~\ref{FigTopology}.  A LEO satellite is equipped with $N=100$, ($N_H = N_V = 10$), antennas and it is located at the position $(300, 300, 400)$~km at our instant of investigation. The antenna
gain at the terrestrial devices is $10.0$~[dBi] and it is $26.9$~[dBi] at the satellite  \cite{bisognin2015millimeter}. The system bandwidth is $B=100$~MHz and the carrier frequency is $f_c = 20$~GHz. The coherence block is $\tau_c = 10000$ OFDM subcarriers. The transmit power assigned to each data symbol is $20$~dBW \cite{abdelsadek2021future}. The noise figure at the APs and satellite are $7$~dB and $1.2$~dB, respectively.  The large-scale fading coefficient between user~$k$ and BS~$m$ is  suggested by the 3GPP model (Release 14) \cite{LTE2017a}, e.g., for a rural area as
%\vspace*{-0.1cm}
\begin{equation}
\beta_{mk} =  G_m + G_k - 8.50 -  20\log_{10}(f_c) -  38.63 \log_{10} (d_{mk})  + \zeta_{mk},
%\vspace*{-0.1cm}
\end{equation}
where $G_m$ and $G_k$ are the antenna gains at AP~$m$ and user~$k$, respectively. the distance between this user and AP~$m$ is denoted as $d_{mk}$ and $\zeta_{mk}$ denotes the shadow fading that follows a log-normal distribution with standard derivation of shadow fading $8$~dB. Meanwhile, the large-scale fading coefficient between user~$k$ and the satellite is defined by using one of the models suggested in \cite{3gpp2019study} as
%\vspace*{-0.1cm}
\begin{equation} \label{eq:betak}
	\beta_{k} = G + G_k + \tilde{G}_k -32.45 - 20 \log_{10} (f_c) - 20 \log_{10} (d_k) + \zeta_k,
%\vspace*{-0.1cm}
\end{equation}
where $G$ is the RA gain at the satellite and its normalized beam pattern is
%\vspace*{-0.1cm}
\begin{equation}
\tilde{G}_k = \begin{cases}
4 \left|J_1 \left( \frac{2\pi}{\lambda} \alpha \sin(\phi_k) \right)/\left(\frac{2\pi}{\lambda} \alpha \sin(\phi_k) \right) \right|^2, & \mbox{if } 0 \leq \phi_k \leq \frac{\pi}{2} ,\\
1, & \mbox{if } \phi_k = 0,
\end{cases}
%\vspace*{-0.1cm}
\end{equation}
where $\beta_k$ denotes the radius of the antenna's circular aperture; $\lambda$ is the wavelength; and $\phi_k$ is the angle between user~$k$ and its beam center. In \eqref{eq:betak}, the shadow fading $\zeta_k$ is obtained from a log-normal distribution with its standard deviation depending on the carrier frequency, channel condition, and the elevation angle \cite{3gpp2019study}. The variable $d_k$~[m] represents the distance between the satellite and user~$k$, defined as
%\vspace*{-0.1cm}
\begin{equation}
	d_k = \sqrt{R_E^2 \sin^2 (\theta_k) + z_0^2 + 2 z_0 R_E} - R_E \sin (\theta_k),
%\vspace*{-0.1cm}
\end{equation}
where $R_E$ is the Earth's radius and $z_0$ is the satellite altitude.  All the numerical results are obtained by a personal Dell Precision $3550$ laptop, $32$~Gb RAM and the CPU Intel(R) Xeon(R) W-$10855$M CPU @ $2.80$~GHz. We consider $1000$ different time slots, each consisting of $20$ users uniformly located in the coverage area. Consequently, there are $20000$ different realizations of the user's locations and shadow fading  in evaluating the system performance.
%\vspace{-0.2cm}
\begin{figure*}[t]
	\begin{minipage}{0.48\textwidth}
		\centering
		\includegraphics[trim=0.8cm 0.0cm 1.1cm 0.6cm, clip=true, width=3.2in]{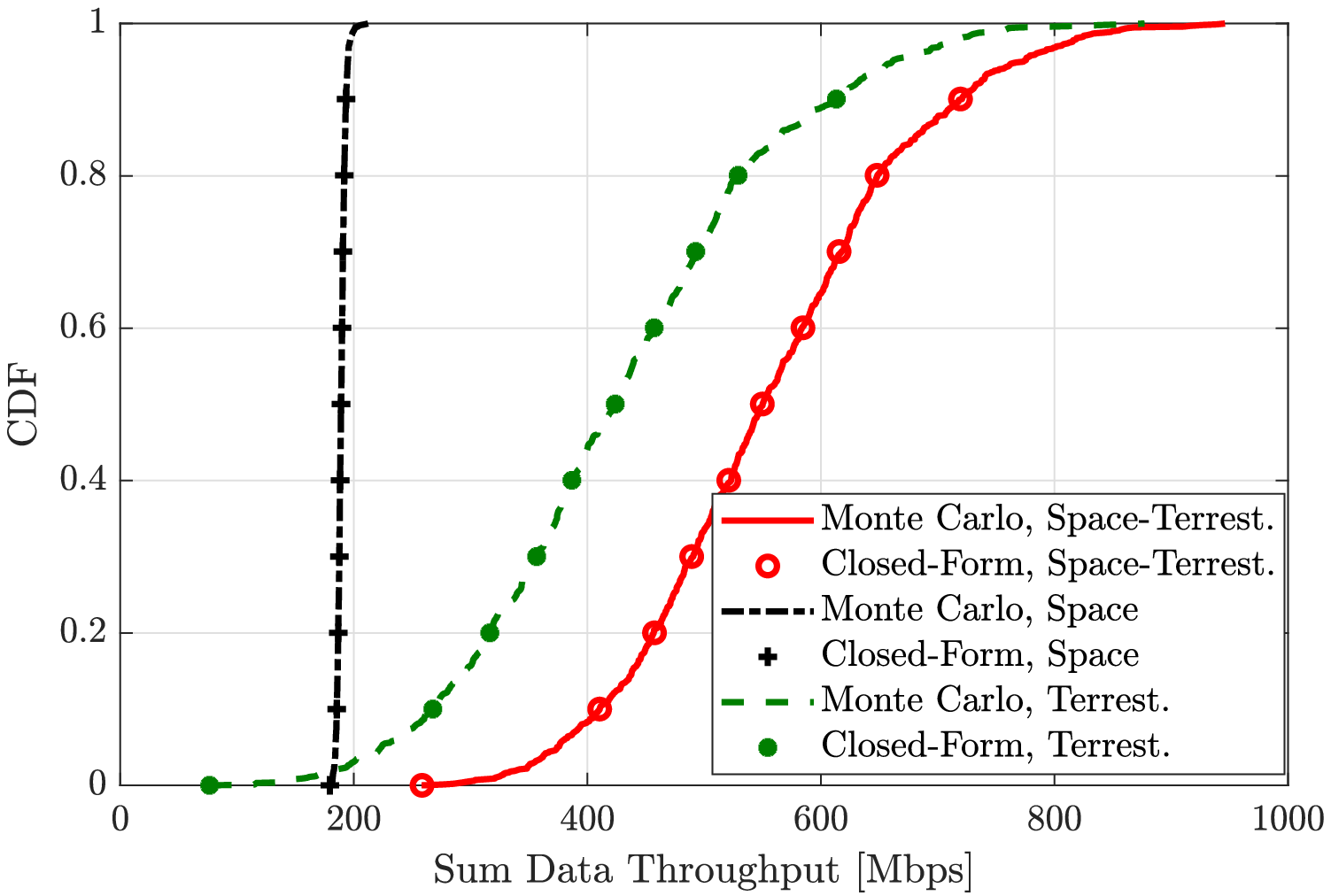} \vspace*{-0.2cm}
		\caption{CDF of the sum ergodic data throughput [Mbps] using Monte Carlo simulations and the analytical frameworks with the MRC technqiue.} \label{Fig:MCCFSUM}
		\vspace*{-0.2cm}
	\end{minipage}
	\hfil
	\begin{minipage}{0.48\textwidth}
		\centering
		\includegraphics[trim=0.8cm 0.0cm 1.4cm 0.6cm, clip=true, width=3.2in]{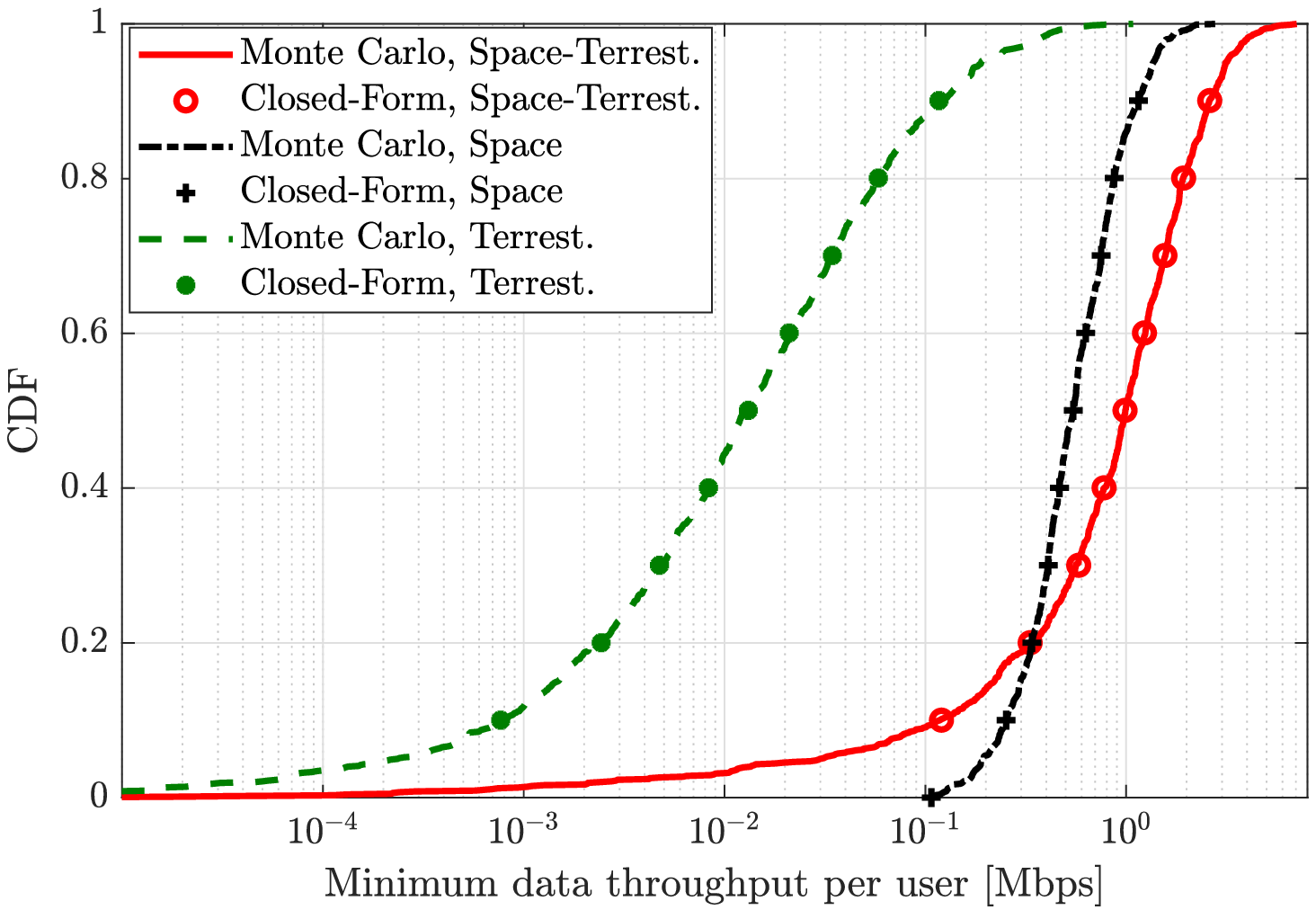} \vspace*{-0.2cm}
		\caption{CDF of the minimum data throughput per user [Mbps] using Monte Carlo simulations vs the analytical frameworks with the MRC technqiue.} \label{Fig:MCCFMINMCCF}
		\vspace*{-0.2cm}
	\end{minipage}
\end{figure*}

\begin{figure*}[t]
	\begin{minipage}{0.48\textwidth}
		\centering
		\includegraphics[trim=0.8cm 0.0cm 1.4cm 0.6cm, clip=true, width=3.2in]{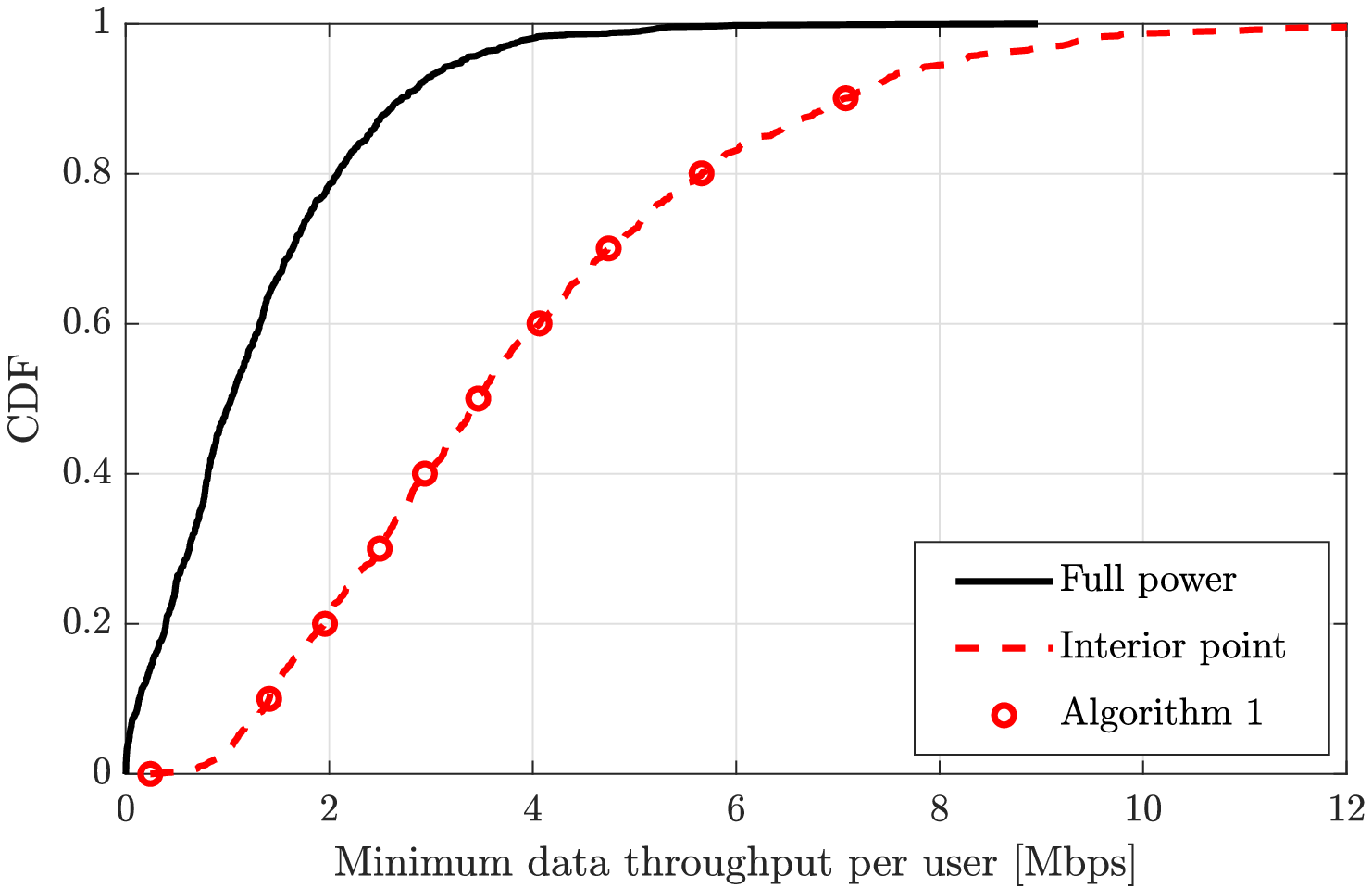} \vspace*{-0.2cm}
		\caption{CDF of the minimum data throughput per user [Mbps] for the space-terrestrial communication system with the MRC technique. } \label{Fig:MCCFMIN}
		\vspace*{-0.2cm}
	\end{minipage}
	\hfil
	\begin{minipage}{0.48\textwidth}
		\centering
		\includegraphics[trim=0.8cm 0cm 1.4cm 0.6cm, clip=true, width=3.2in]{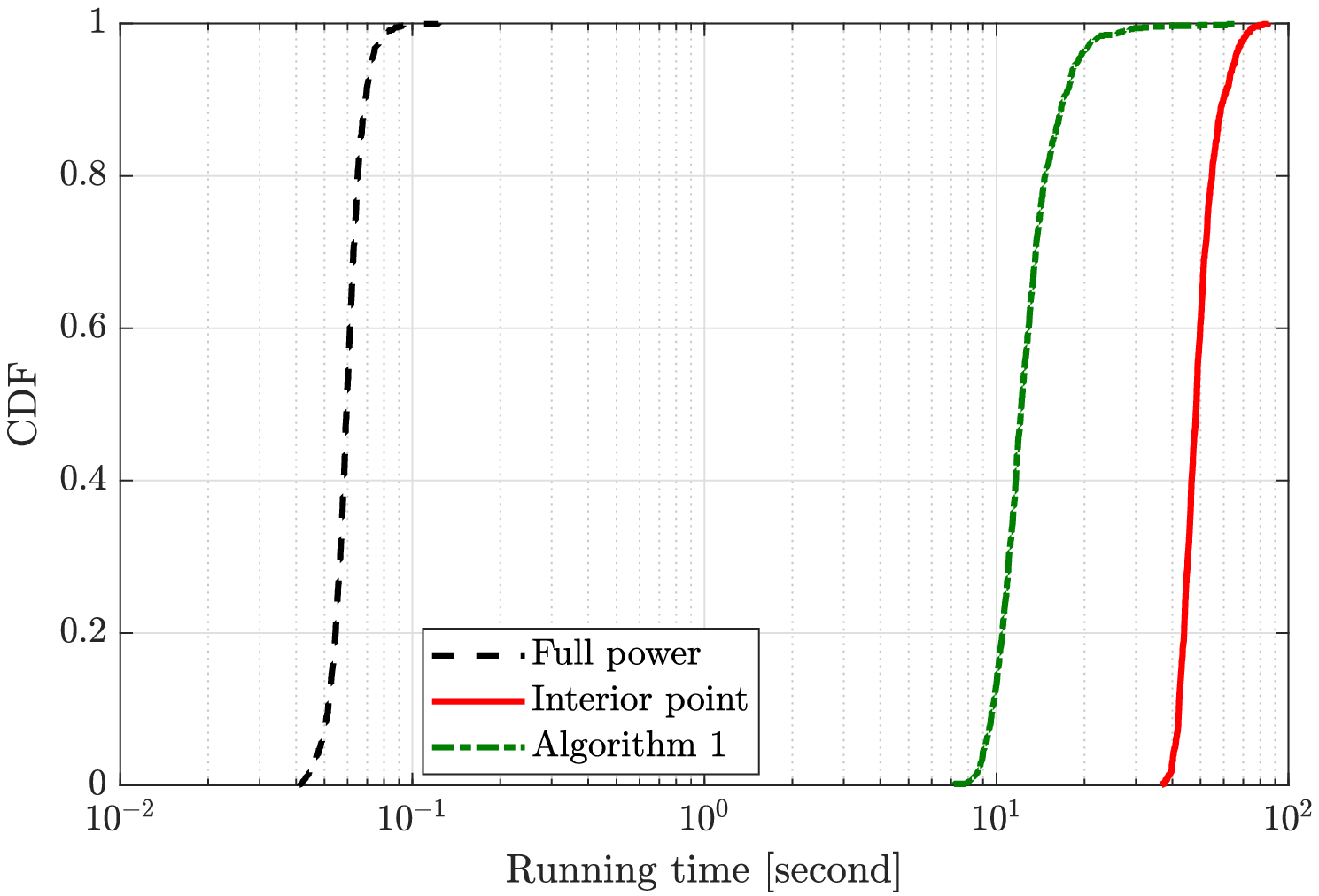}  \vspace*{-0.2cm}
		\caption{CDF of the running time to obtain the max-min fairness solution for the space-terrestrial communication system with MRC technique. }  \label{Fig:MCCFTIME}
		\vspace*{-0.2cm}
	\end{minipage}
\end{figure*}
\begin{figure*}[t]
	\begin{minipage}{0.48\textwidth}
		\centering
		\includegraphics[trim=0.8cm 0cm 1.4cm 0.6cm, clip=true, width=3.2in]{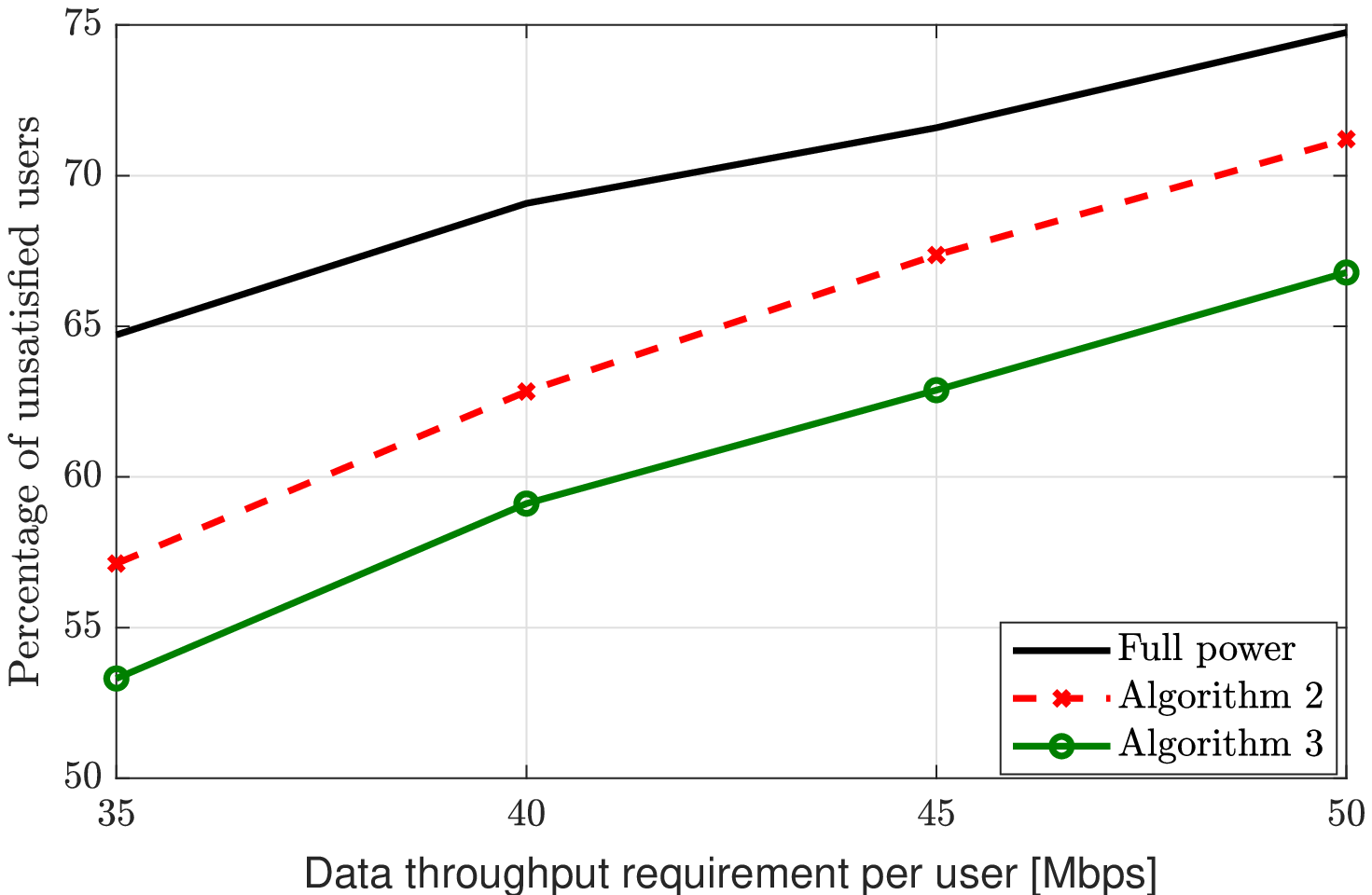}  \vspace*{-0.2cm}
		\caption{The percentage of unsatisfied users for the space-terrestrial communication system with the MRC technique.} \label{Fig:UnsatisfiedUsers}
		\vspace*{-0.2cm}
	\end{minipage}
	\hfil
	\begin{minipage}{0.48\textwidth}
		\centering
		\includegraphics[trim=0.8cm 0cm 1.4cm 0.6cm, clip=true, width=3.2in]{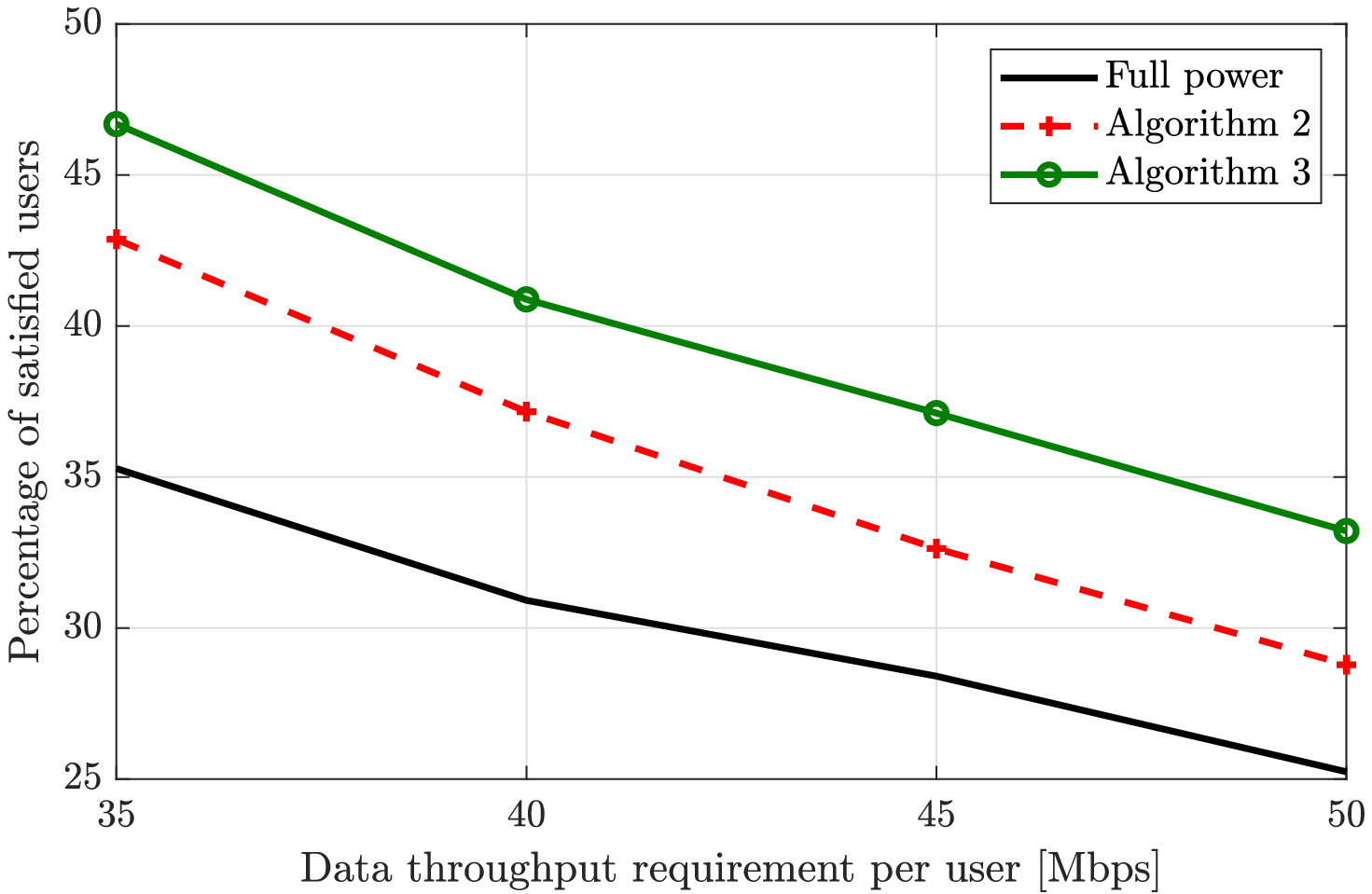} \vspace*{-0.2cm}
		\caption{The percentage of satisfied users for the space-terrestrial communication system with the MRC technique.} \label{Fig:SatisfiedUsers}
		\vspace*{-0.2cm}
	\end{minipage}
\end{figure*}
\begin{figure*}[t]
	\begin{minipage}{0.48\textwidth}
		\centering
		\includegraphics[trim=0.7cm 0cm 1.4cm 0.6cm, clip=true, width=3.2in]{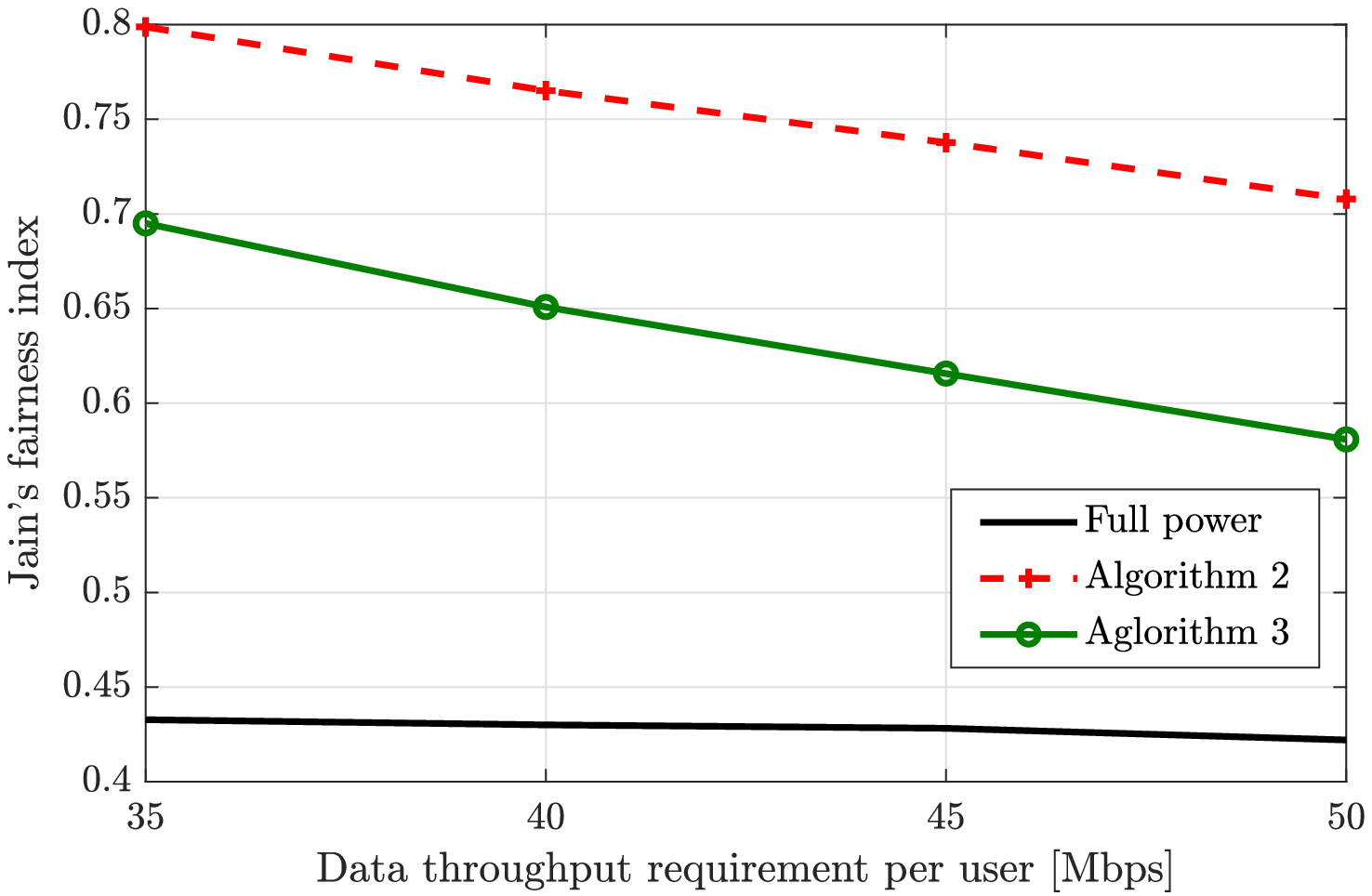} \vspace*{-0.2cm}
		\caption{The Jain's fairness index for the space-terrestrial communication system with the MRC technique.} \label{Fig:Jain}
		\vspace*{-0.2cm}
	\end{minipage}
	\hfil
	\begin{minipage}{0.48\textwidth}
		\centering
		\includegraphics[trim=0.8cm 0cm 1.4cm 0.6cm, clip=true, width=3.2in]{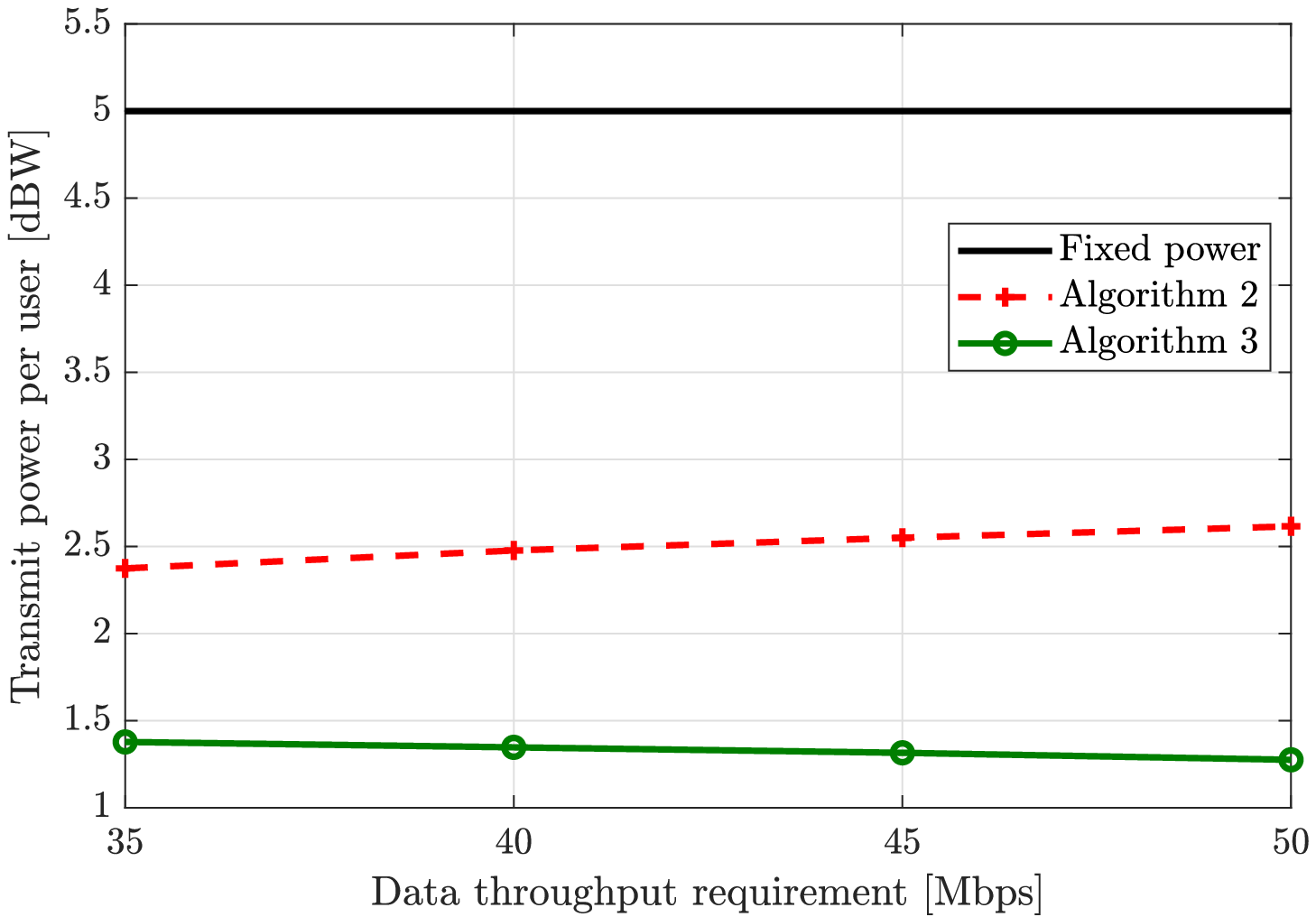} \vspace*{-0.2cm}
		\caption{The transmit power per user [dBW] for the space-terrestrial communication system with the MRC technique. } \label{Fig:TransmitPower}
		\vspace*{-0.2cm}
	\end{minipage}
\end{figure*}
\begin{figure*}[t]
	\begin{minipage}{0.48\textwidth}
	\centering
	\includegraphics[trim=0.6cm 0cm 1.0cm 0.6cm, clip=true, width=3.2in]{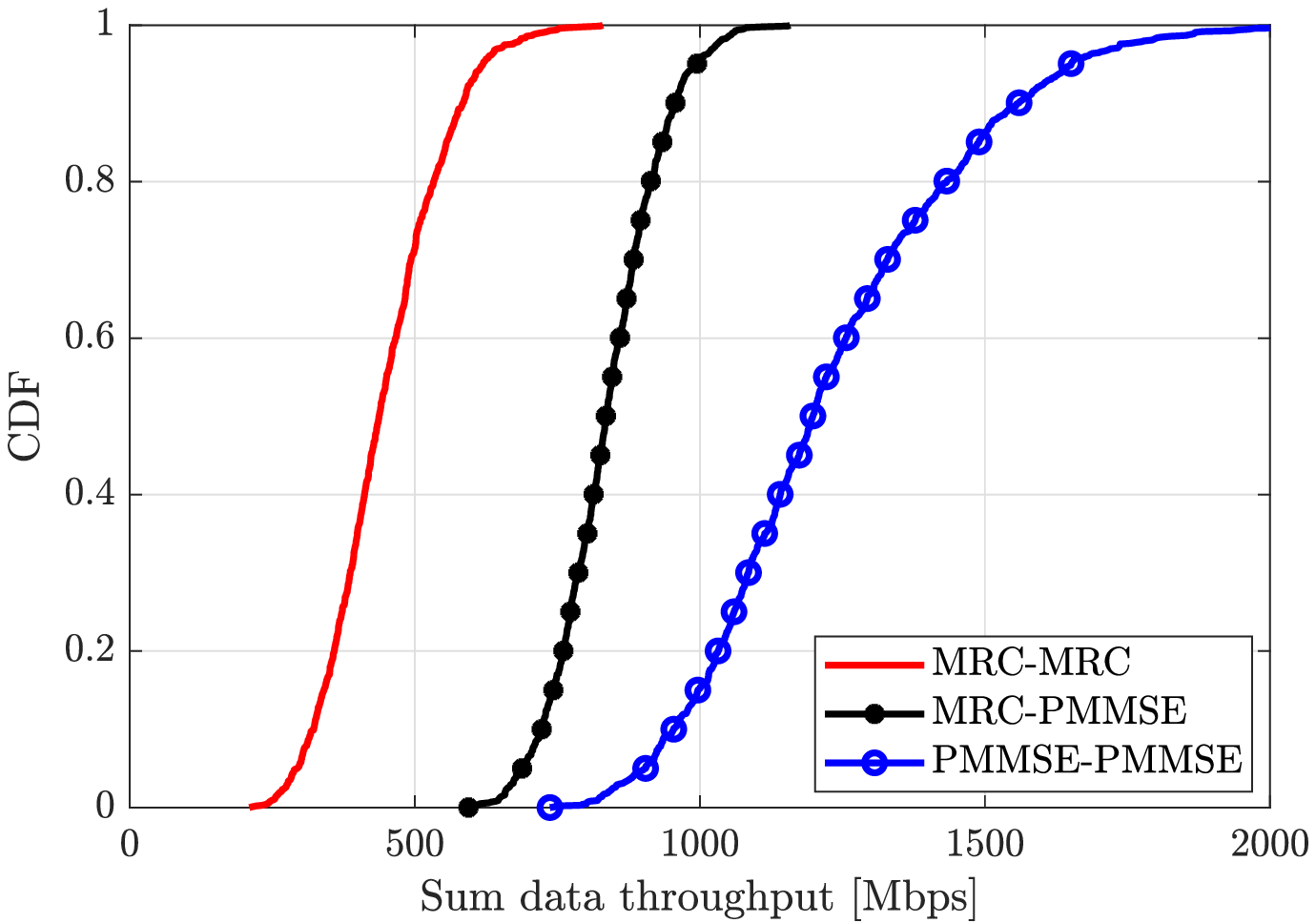} \vspace*{-0.2cm}
	\caption{The CDF of the sum data throughput per user [Mbps] with different combining methods.} \label{Fig:DiffCom}
	\vspace*{-0.2cm}
	\end{minipage}
\hfil
    \begin{minipage}{0.48\textwidth}
    	\centering
    	\includegraphics[trim=0.6cm 0cm 1.0cm 0.6cm, clip=true, width=3.2in]{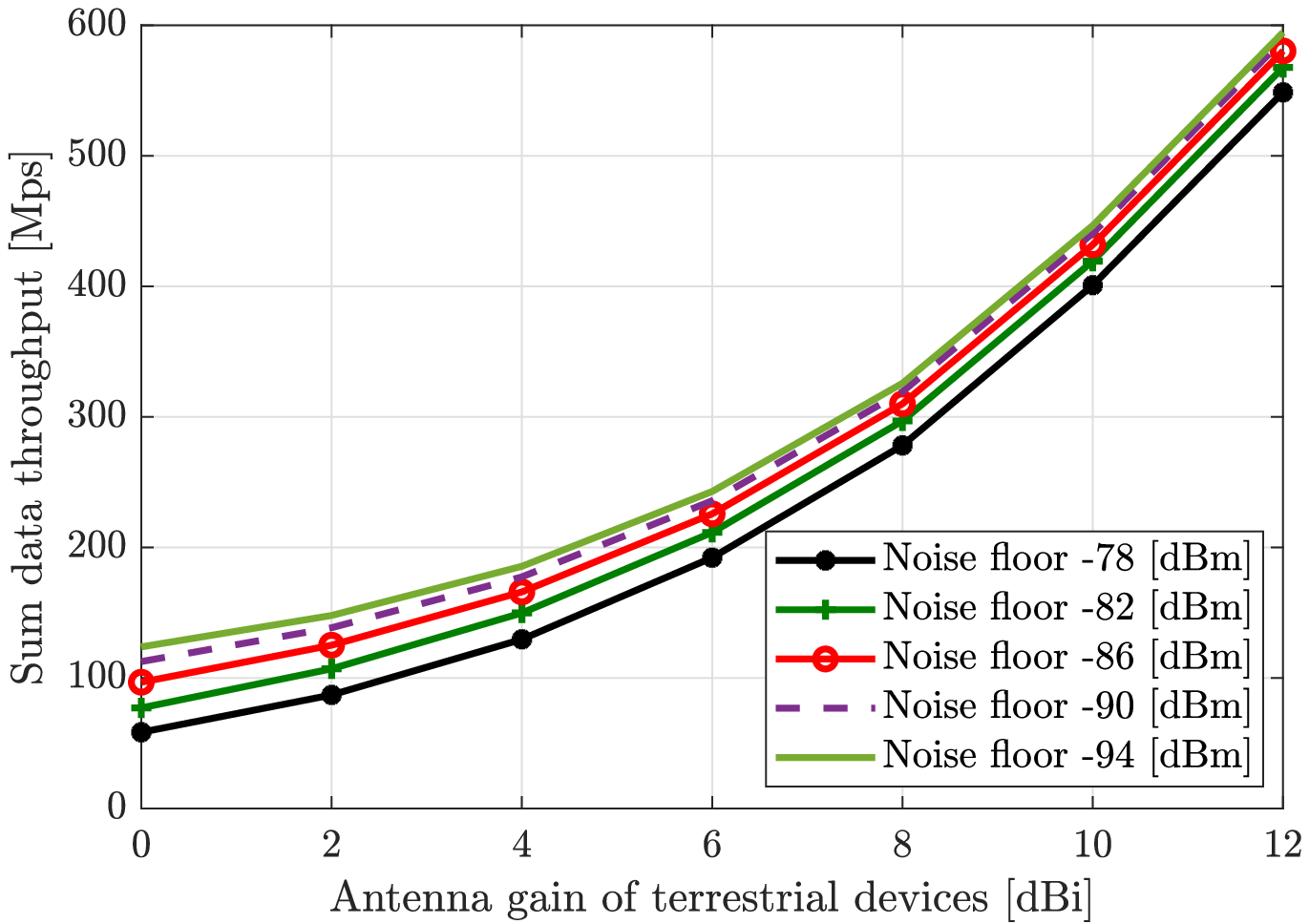} \vspace*{-0.2cm}
    	\caption{The sum data throughput [Mbps] versus the antenna gain of terrestrial devices [dB].} \label{Fig:DiffNoiseFloor}
    	\vspace*{-0.2cm}
    \end{minipage}
\end{figure*}

\subsection{Analytical Results versus Monte-Carlo Simulations}
In Fig.~\ref{Fig:MCCFSUM}, we compare the cumulative distribution function (CDF) of the sum data throughput, defined as $\sum_{k=1}^K R_k$, between Monte Carlo simulations and the proposed analytical framework for the three different systems: $i)$ \textit{the  space-terrestrial communication system} represented by the SINR expression in \eqref{eq:SINRk} with the overall channel coefficient $z_{kk'}$ in \eqref{eq:zkkprime} and the analytical framework in \eqref{eq:ClosedSINR}; $ii)$ \textit{the stand-alone terrestrial communication system} represented by the SINR expression in \eqref{eq:SINRk} with $z_{kk'}=   \sum\nolimits_{m=1}^M  u_{mk}^\ast  g_{mk'} $ and the analytical framework in \eqref{eq:SINRTerrest}; and $iii)$ \textit{the stand-alone space communication system} represented by the SINR expression in \eqref{eq:SINRk} with $z_{kk'}=  \mathbf{u}_k^H \mathbf{g}_{k'} $ and the analytical framework in \eqref{eq:SINRSatelliteOnly}. The numerical simulations and the analytical results match very well for the different systems that validates our analysis. A network only relying on the satellite offers about $189$~[Mbps] throughput on average, which is quite stable over different user locations and shadow fading coefficients. The terrestrial system offers $2.3$ times better the sum throughput than the baseline. Jointly processing the received signals, the space-terrestrial communication system supports superior improvements of $30\%$ in throughput over utilizing the APs only.

In Fig.~\ref{Fig:MCCFMINMCCF}, we show the CDF of the minimum throughput of the system, which is defined as $R_k^{\min} = \underset{k}{\min} \, R_k$ for the three systems considered as in line with Fig.~\ref{Fig:MCCFSUM}. The terrestrial communication system is the baseline for the minimum  throughput, which only provides $0.05$~[Mbps] per user on average. The satellite communication system yields a minimum data throughput of about $0.63$~[Mbps] per user, which is $14$ times higher than the baseline. Superior gains up to a factor of $28.8$ times better than only using the APs are obtained by integrating the satellite into a terrestrial network. At $95\%$-likelihood, the space communication system offers the best minimum data throughput of about $0.21$~[Mbps], which is $6.5$ times and $886$ times better than the space-terrestrial and  terrestrial communication system, respectively.

\subsection{Power Control versus Fixed Power}
In Fig.~\ref{Fig:MCCFMIN}, we plot the CDF of the data throughput of the different power allocation strategies. The full data power transmission used as a popular benchmark \cite{9460776} produces the lowest average max-min fairness level, which is $1.3$~[Mbps]. The two remaining algorithms generate the same solution that is $3\times$ better than the full data power transmission on average. This observation demonstrates significant enhancements of the max-min fairness power allocation for the users having low channel quality. Indeed, the interior point methods have been widely applied for solving the max-min fairness optimization problem \cite{6519406}. The associated running time required to obtain the max-min fairness solution is given in Fig.~\ref{Fig:MCCFTIME}. The running time is only $0.06$~[s] dedicated to estimating the data throughput if the system allows each user transmit at full power per data symbol. Carefully optimizing the data powers requires extra computational complexity for the iterative processes. The interior-point methods spend $50$~[s] to obtain the solution of  Problem~\eqref{Problem:MaxMinQoS}, while Algorithm~\ref{Algorithm1} only needs $13$~[s] corresponding to a rich time reduction factor of $3.8 \times$. The result verifies the benefits of the proposed power update based on the standard interference function in solving the max-min fairness optimization problem. 

In Fig.~\ref{Fig:UnsatisfiedUsers}, we portray the percentage of user locations and shadow fading coefficients yielding a  throughput lower than requested (unsatisfied users).\footnote{We only use the full power consumption scenario as our  benchmark for comparison, since this is the first time in the satellite-terrestrial literature that congestion control has been taken into account. More explicitly, other algorithmic designs such as \cite{9460776,van2021user} are not included for comparison, since they were developed for non-empty feasible sets that satisfy Slater's condition and could not handle the congestion issue. These algorithms always give infeasible solutions under the data throughput requirements considered.} All the benchmarks show an increasing trend of unsatisfied users, as the data throughput requirement increases. The full power allocation yields the highest percentage of unsatisfied users that varies from about $65\%$ to approximately $75\%$, when the users change the requested data throughput from $35$~[Mbps] to $50$~[Mbps]. The two proposed algorithms smoothly handle the congestion control.  Algorithm~\ref{Algorithm2} allows each unsatisfied user to transmit at full power that may inflict severe mutual interference upon the remaining users. Therefore, the percentage of unsatisfied users ranges from $57\%$ to $71\%$, depending on the throughput requested. By carefully reducing the transmit power  of unsatisfied users and therefore managing the mutual interference efficiently, Algorithm~\ref{Algorithm3} has the lowest percentage of unsatisfied users, which is from $53\%$ to $67\%$. As a further result, Fig.~\ref{Fig:SatisfiedUsers} illustrates the percentage of satisfied users as a function of the data throughput requirements. By softly removing unsatisfied users, Algorithm~\ref{Algorithm3} offers the highest percentage of satisfied users, followed by Algorithm~\ref{Algorithm2} and the full power transmission.  

As for the network-wide fairness, Fig.~\ref{Fig:Jain} evaluates  Jain's fairness index as defined in \eqref{eq:JainFairnessIndex}. Without making use of the spatial diversity and channel statistics, the full power transmission gives the worst Jain fairness index. Interestingly, Algorithm~\ref{Algorithm2}  provides the highest Jain fairness index,  best supporting each unsatisfied user. Even though Algorithm~\ref{Algorithm3} helps to increase the number of satisfied users, this is at the cost of degrading the throughput of the unsatisfied users. Consequently, the network-wise fairness may not be the best, as reported in Fig.~\ref{Fig:Jain}.  We also show the data power per symbol in Fig.~\ref{Fig:TransmitPower}. Both the proposed algorithms consume significantly less power than the maximum power level. In particular, Algorithm~\ref{Algorithm2} reduces the power consumption up to $2.1\times$, and that of Algorithm~\ref{Algorithm3} is $3.9\times$. The results quantify the energy efficiency of these algorithms under congestion.  
 
\subsection{Other Observations}
The CDF of the sum data throughput [Mbps] is displayed in Fig.~\ref{Fig:DiffCom} by using either the partial MMSE (P-MMSE)  or the MRC receiver for detecting the desired signals at both the space and terrestrial links. The P-MMSE matrix $\mathbf{U}_s = [\mathbf{u}_{1s}, \ldots, \mathbf{u}_{Ks}] \in \mathbb{C}^{N \times K}$ of an MMSE receiver is configured for the space link as
	%\begin{equation}
	$	\mathbf{U}_s = \widehat{\mathbf{G}} \big( \widehat{\mathbf{G}}^H \widehat{\mathbf{G}} + K\sigma_s^2\mathbf{I}_K /P_{\max} \big)^{-1}$,
	%\end{equation}
	where we have $\widehat{\mathbf{G}} = [\hat{\mathbf{g}}_1, \ldots, \hat{\mathbf{g}}_K] \in \mathbb{C}^{N \times K}$ and $P_{\max} = P_{\max,k}, \forall k$. By contrast, the P-MMSE matrix is formulated as $\mathbf{U}_p = \widehat{\mathbf{H}}\big( \widehat{\mathbf{H}}^{H} \widehat{\mathbf{H}} + K \sigma_s^2 \mathbf{I}_K /P_{\max} \big)^{-1} \in \mathbb{C}^{M \times K}$, where the $(m,k)$-th element of the matrix $\widehat{\mathbf{H}}$ is defined as $[\widehat{\mathbf{H}}]_{mk} = \hat{g}_{mk}, \forall m,k$. 
	Compared to the MRC receiver, significant improvements are attained by using the P-MMSE solution for detecting the received signals. Specifically, the improvement is by a factor of about $2.1$ on average, if the P-MMSE combiner is used by the CPU. Furthermore, if the CPU utilizes the P-MMSE combiner for both the space and terrestrial links, the sum data throughput improvement is by a factor of about $2.5$ on average. The results reveal the benefits of our linear receiver combiner designed for supporting the collaboration of the satellite and APs by exploiting the  associated channel estimates.

In Fig.~\ref{Fig:DiffNoiseFloor}, we plot the sum data throughput [Mbps] as a function of the antenna gain at each terrestrial device for different noise floors. The terrestrial devices are equipped with omnidirectional antennas having a gain of $0$~[dB]. Furthermore, we assume that the users are uniformly scattered throughout the coverage area. The higher antenna gains offer significantly better sum data throughput, but sophisticated beam-search techniques must be used to detect the radio beams \cite{9846951}. Besides, the different noise floors characterize the imperfect feeder link and imperfect synchronization between the satellite and terrestrial links.

%\vspace*{-0.4cm}
\section{Conclusions} \label{Sec:Concl}
%\vspace*{-0.3cm}
The throughput analysis and the data power control of a multi-user system were provided in the presence of an NGSO satellite and distributed APs for improving the macro-diversity gains attained. We assumed a centralized signal processing unit for boosting the data throughput per user for coherent data detection combining the received signals of the space and terrestrial links. The achievable data throughput expression derived can be applied to an arbitrary channel model and combining techniques. A closed-form expression was also derived for the MRC receiver technique and spatially correlated  channels with a rich scattering environment around users. 
The satellite boosts the sum throughput in the network by more than $30\%$ for the  parameter setting considered, while the minimum data throughput is enhanced by more than tenfold. Two different optimization problems were formulated to study the data power allocation on a long-term scale, where the solution is only updated whenever the channel statistics vary. Our Monte Carlo simulations demonstrate that many users can still access the network and attain satisfactory  throughput under congested conditions.
%\vspace*{-0.3cm}
\appendix
%\vspace*{-0.4cm}
\subsection{A Useful Lemma and Definitions}
This appendix presents the following lemma and definitions for our throughput analysis and optimization.
%\vspace*{-0.2cm}
\begin{lemma}\cite[Lemma~4]{van2018large} \label{lemma:4moment}
If a random vector $\mathbf{x} \in \mathbb{C}^{N}$ is distributed as $\mathbf{x}  \sim \mathcal{CN}(\mathbf{0}, \mathbf{R})$ where $\mathbf{R} \in \mathbb{C}^{N \times N}$ denotes the covariance matrix, the following property holds for an arbitrary deterministic matrix $\mathbf{N} \in \mathbb{C}^{N \times N}$:
%\vspace*{-0.1cm}
\begin{equation}
\mathbb{E}\{ |\mathbf{x}^H \mathbf{N} \mathbf{x}|^2 \} = |\mathrm{tr}( \mathbf{R} \mathbf{N} ) |^2 + \mathrm{tr}(\mathbf{R}\mathbf{N} \mathbf{R} \mathbf{N}^H).
\vspace*{-0.1cm}
\end{equation}
\end{lemma}
%\vspace*{-0.2cm}
\begin{definition} \cite[Definition~3]{van2021uplink}  \label{Def:SIF}
A function $I(\mathbf{x})$ is a standard interference function, if the following properties hold
\begin{itemize}
	\item[$i)$] Positivity: For all $\mathbf{x} \succeq \mathbf{0}$, $I(\mathbf{x}) > 0$.
	\item[$ii)$] Monotonicity: For two vectors $\mathbf{x}$ and $\hat{\mathbf{x}}$ satisfied $\hat{\mathbf{x}} \succeq \mathbf{x}$,  $I(\hat{\mathbf{x}}) \geq I(\mathbf{x})$.
	\item[$iii)$] Scalability: For all constant values $\alpha > 1$, $\alpha I(\mathbf{a}) > I (\alpha \mathbf{a})$. 
\end{itemize}
%\vspace*{-0.2cm}
\end{definition}
\begin{definition} \cite[Section~IV]{sung2005generalized} \label{Def:TSSF}
For a given $\alpha > 1$ and a pair of vectors $\mathbf{x}$ and $\hat{\mathbf{x}}$ satisfied $\frac{1}{\alpha} \mathbf{x} \preceq \hat{\mathbf{x}} \preceq \alpha \mathbf{x}$, a function $f(\mathbf{x})$ is a two-sided scalable function if the following property holds
%\vspace*{-0.1cm}
\begin{equation}
f(\mathbf{x})/\alpha < f(\hat{\mathbf{x}}) < \alpha f(\mathbf{x}).
\vspace*{-0.1cm}
\end{equation}
\end{definition}
%\vspace*{-0.5cm}
\subsection{Proof of Theorem~\ref{Theorem:ClosedForm}} \label{Appendix:ClosedForm}
%\vspace*{-0.2cm}
By utilizing the overall channel based on its definition with $k' = k$, the numerator of \eqref{eq:SINRk} is formulated as
%\vspace*{-0.1cm}
\begin{equation} \label{eq:zkkv1}
\begin{split}
&|\mathbb{E} \{z_{kk}\}|^2 = \big|\|\bar{\mathbf{g}}_k\|^2 +  p K \mathrm{tr}(\pmb{\Theta}_k)  +  \sum\nolimits_{m=1}^M  \gamma_{mk} \big|^2,
\end{split}
%\vspace*{-0.1cm}
\end{equation}
which is obtained for the channel distributions considered and the MRC technique. We denote the first term in the denominator of \eqref{eq:SINRk}, $\mathsf{D}_1 = \sum\nolimits_{k'=1}^K \rho_{k'}  \mathbb{E}\{ |z_{kk'}|^2 \}$, which is expressed as
%\vspace*{-0.1cm}
\begin{equation} \label{eq:Denov1}
 \mathsf{D}_1 =  \rho_{k}  \mathbb{E}\{ |z_{kk}|^2 \}  +  \mathsf{D}_{2}, 
 %\vspace*{-0.1cm}
\end{equation}
where $ \mathsf{D}_{2} = \sum\nolimits_{k' = 1, k' \neq k}^K \rho_{k'}  \mathbb{E}\{ |z_{kk'}|^2$. $\mathbb{E}\{ | z_{kk} |^2 \}$ is tackled  as follows
%\vspace*{-0.1cm}
\begin{equation} \label{eq:zkkGain}
\begin{split}
 &\mathbb{E} \{ |z_{kk}|^2 \} = \mathbb{E} \{ | a_{kk}  + \tilde{a}_{kk} + b_{kk} + \tilde{b}_{kk} |^2 \} \\
 & = \mathbb{E} \{ | a_{kk} |^2 \}  +  \mathbb{E} \{ | \tilde{a}_{kk} |^2 \} + \mathbb{E} \{ | b_{kk} |^2 \} + \mathbb{E} \{ | \tilde{b}_{kk} |^2 \} \\
 & \, \, \, + 2\mathbb{E}\{ a_{kk} b_{kk} \},
\end{split}
%\vspace*{-0.1cm}
\end{equation}
where $a_{kk} =  \|\hat{\mathbf{g}}_{k}\|^2$, $\tilde{a}_{kk} = \hat{\mathbf{g}}_k^H \mathbf{e}_{k}$, $b_{kk} = \sum_{m=1}^M  |\hat{g}_{mk}|^2$, and $\tilde{b}_{kk} = \sum_{m=1}^M \hat{g}_{mk}^\ast e_{mk}$. In \eqref{eq:zkkGain}, the remaining expectations vanish due to the zero mean of the AWGN noise. The first expectation $\mathbb{E} \{ | a_{kk} |^2 \}$ in \eqref{eq:zkkGain} is given by
%\vspace*{-0.1cm}
\begin{equation} \label{eq:hatgkgk}
	 \mathbb{E} \{ | a_{kk} |^2 \}
	 =  \big( \|\bar{\mathbf{g}}_{k}\|^2 + 2p K \mathrm{tr}(\pmb{\Theta}_k)\big)^2  + 2 p K \bar{\mathbf{g}}_k^H \pmb{\Theta}_k  \bar{\mathbf{g}}_k +  p^2 K^2 \mathrm{tr} (\pmb{\Theta}_k^2).
%\vspace*{-0.1cm}
\end{equation}
We can derive the closed-form expression of the second expectation $\mathbb{E} \{ | \tilde{a}_{kk} |^2 \}$ in \eqref{eq:zkkGain} as follows
\begin{multline} \label{eq:ekhatgk}
		\mathbb{E} \{ | \tilde{a}_{kk} |^2 \} = \bar{\mathbf{g}}_{k}^H \mathbf{R}_{k} \bar{\mathbf{g}}_{k} - p K \bar{\mathbf{g}}_{k}^H  \pmb{\Theta}_k  \bar{\mathbf{g}}_{k} +  p K \mathrm{tr}( \mathbf{R}_{k} \pmb{\Theta}_k  ) \\ -  p^2 K^2 \mathrm{tr}( \pmb{\Theta}_k^2 ),
\end{multline}
which is obtained by utilizing the covariance matrix of the channel estimation error in \eqref{eq:EstSat}. The third expectation $\mathbb{E} \{ | b_{kk} |^2 \}$ in the last equality of \eqref{eq:zkkGain} is formulated as follows
%	\vspace*{-0.1cm}
\begin{equation}
	\begin{split}
		&\mathbb{E} \{ | b_{kk} |^2 \} = \sum\nolimits_{m=1}^M  \gamma_{mk}^2 + \left( \sum\nolimits_{m=1}^M   \gamma_{mk} \right)^2,
	\end{split}
%	\vspace*{-0.1cm}
\end{equation}
by applying  Lemma~\ref{lemma:4moment} and the channel estimates in Lemma~\ref{Lemma:Est}. The fourth expectation $\mathbb{E} \{ | \tilde{b}_{kk'} |^2 \}$ in the last equality of \eqref{eq:zkkGain} is expressed as follows
%	\vspace*{-0.1cm}
\begin{equation}
\mathbb{E} \{ | \tilde{b}_{kk} |^2 \} =   \sum\nolimits_{m=1}^M    \gamma_{mk} (\beta_{mk} - \gamma_{mk}),
%	\vspace*{-0.1cm}
\end{equation}
as a consequence of the mutual independence between the channel estimate and estimation error. The fifth expectation $\mathbb{E}\{ a_{kk}b_{kk} \} $  in the last equality of \eqref{eq:zkkGain} is expanded as follows
%	\vspace*{-0.1cm}
\begin{equation} \label{eq:6thterm}
\begin{split}
& \mathbb{E}\{ a_{kk} b_{kk} \} = \big( p K \mathrm{tr} ( \pmb{\Theta}_k) + \|\bar{\mathbf{g}}_k\|^2 \big) \sum\nolimits_{m=1}^M   \gamma_{mk},
\end{split}
%	\vspace*{-0.1cm}
\end{equation}
thanks to the channel statistics in Lemma~\ref{Lemma:Est}. Substituting \eqref{eq:hatgkgk}--\eqref{eq:6thterm} into \eqref{eq:zkkGain}, we obtain the closed-form expression of $\mathbb{E}\{ |z_{kk'}|^2\}$. That of $\mathsf{MI}_1$ is given by
\begin{multline}  \label{eq:zkk}
\mathbb{E}\{ |z_{kk}|^2 \} =  \left( \|\bar{\mathbf{g}}_{k} \|^2  +  \mathrm{tr}(\pmb{\Theta}_k) +   \sum\limits_{m=1}^M   \gamma_{mk} \right)^2 + p K \bar{\mathbf{g}}_{k}^H \pmb{\Theta}_k \bar{\mathbf{g}}_{k} \\ + \bar{\mathbf{g}}_{k}^H \mathbf{R}_{k} \bar{\mathbf{g}}_{k}  +  p K \mathrm{tr}( \mathbf{R}_{k} \pmb{\Theta}_k ) +  \sum\limits_{m=1}^M \gamma_{mk} \beta_{mk}.
\end{multline}
Next, the mutual interference $\mathsf{D}_2$ in \eqref{eq:Denov1} is handled as follows
%	\vspace*{-0.1cm}
\begin{equation} \label{eq:D2}
\begin{split}
& \mathsf{D}_2 = \sum\nolimits_{k'=1, k' \neq k }^K \rho_{k'} \left(  \mathbb{E}\{ | \hat{\mathbf{g}}_k^H \mathbf{g}_{k'}|^2  \} +   \sum\nolimits_{m=1}^M   \mathbb{E}\{  | \hat{g}_{mk}^\ast  g_{mk'}|^2 \} \right)\\
&=  p K \sum\nolimits_{k' =1, k' \neq k }^K \rho_{k'}   \mathrm{tr}( \mathbf{R}_{k'}\pmb{\Theta}_k )  + p K \sum\nolimits_{k' =1, k' \neq k }^K \rho_{k'}  \bar{\mathbf{g}}_{k'}^H  \pmb{\Theta}_k   \bar{\mathbf{g}}_{k'} \\
& +  \sum\nolimits_{k' =1, k'\neq k }^K \rho_{k'} \bar{\mathbf{g}}_{k}^H \mathbf{R}_{k'}  \bar{\mathbf{g}}_{k}  +    \sum\nolimits_{k' =1, k' \neq k }^K \rho_{k'} |\bar{\mathbf{g}}_{k}^H  \bar{\mathbf{g}}_{k'} |^2 \\
& + \sum_{k' =1, k' \neq k}^K \sum_{m=1}^M \rho_{k'}  \gamma_{mk} \beta_{mk'},
\end{split}
%	\vspace*{-0.1cm}
\end{equation}
thanks to the mutual independence of the channels impinging from users utilizing the orthogonal pilot signals.
The noise power from the satellite communication section is represented by in the closed-form expression of
%	\vspace*{-0.1cm}
\begin{equation} \label{eq:NoiseS}
\mathbb{E} \big\{ \big| \hat{\mathbf{g}}_k^H \mathbf{w} \big|^2 \big\} \stackrel{(a)}{=} \mathbb{E} \big\{ \big| \hat{\mathbf{g}}_k^H  \mathbb{E}\{ \mathbf{w} \mathbf{w}^H \} \hat{\mathbf{g}}_k \big|^2 \big\} = \sigma_s^2 \|\bar{\mathbf{g}}_k\|^2 +  p K\sigma_s^2 \mathrm{tr}( \pmb{\Theta}_k), 
%	\vspace*{-0.1cm}
\end{equation}
where $(a)$ is obtained by the independence of the channel estimate and noise. Similarly, the noise power from the APs is computed in the closed-form expression as
%	\vspace*{-0.1cm}
\begin{equation} \label{eq:NoiseA}
\sum\nolimits_{m=1}^M  \mathbb{E} \big\{ | \hat{g}_{mk}^\ast w_m |^2 \big\} = \sigma_a^2 \sum\nolimits_{m=1}^M   \mathbb{E} \big\{ | \hat{g}_{mk} |^2 \big\} =  \sigma_a^2 \sum\nolimits_{m=1}^M  \gamma_{mk}.
%	\vspace*{-0.1cm}
\end{equation}
Substituting \eqref{eq:zkk} and \eqref{eq:D2} into \eqref{eq:Denov1}, we obtain the closed-form expression of the first term in the denominator of \eqref{eq:SINRk}. Then, this result together with \eqref{eq:zkkv1}, \eqref{eq:NoiseS}, and \eqref{eq:NoiseA} leads to the closed-form expression of the throughput as in the theorem.
%\vspace*{-0.4cm}
\subsection{Proof of Lemma~\ref{Lemma:QuasiConcave}} \label{Appendix:QuasiConcave}
%\vspace*{-0.2cm}
Based on the relationship between the throughput and SINR value in \eqref{eq:Rkv1}, we move from the data throughput optimization problem of \eqref{Problem:MaxMinQoS} to the corresponding weakest SINR maximization problem formulated as follows
%	\vspace*{-0.1cm}
\begin{equation} \label{Problem:MaxMinSINR}
	\begin{aligned}
		& \underset{ \{ \rho_{k} \} }{\textrm{maximize}} \; \underset{k}{\textrm{min}}
		& &  \mathrm{SINR}_k \\
		& \textrm{subject to}
		& & 0 \leq \rho_{k} \leq P_{\mathrm{max},k} \;, \forall k  ,\\
	\end{aligned}
%	\vspace*{-0.1cm}
\end{equation}
Let us introduce the set $\mathcal{X} = \{ \rho_k \}$ that contains all the data power variables. The objective function of problem~\eqref{Problem:MaxMinSINR} is defined as
%	\vspace*{-0.1cm}
\begin{equation} \label{eq:f0S}
f_0(\mathcal{S}) = \underset{k}{\textrm{min}} \, \, \mathrm{SINR}_k.
%	\vspace*{-0.1cm}
\end{equation}
For any $\xi >0$ representing a lower bound of the SINR values, the upper level set of the function $f_0(\mathcal{S})$ is formulated as
%	\vspace*{-0.1cm}
\begin{equation} \label{eq:UpLSet}
\begin{split}
& U(\mathcal{S}, \xi) = \{ \mathcal{S}| f_0(\mathcal{S}) \geq \xi\}  = \left\{ \mathcal{S} \Big| \frac{\xi \mathsf{MI}_k(\mathcal{S})}{\rho_k a_k} + \frac{\xi \mathsf{NO}_k}{\rho_k a_k}\leq 1, \forall k \right\},
\end{split}
%	\vspace*{-0.1cm}
\end{equation}
where we have $a_k =  | \|\bar{\mathbf{g}}_k\|^2 +  p K \mathrm{tr}(\pmb{\Theta}_k)  +  \sum_{m=1}^M \gamma_{mk} |^2$. Notice that the last equality of \eqref{eq:UpLSet} is obtained by substituting the closed-form expression \eqref{eq:ClosedSINR} into \eqref{eq:f0S} and then applying some algebraic manipulations. The upper-level set $U(\mathcal{S}, \xi)$ is a convex set by a logarithmic change of the optimization variables since it consists of the posynomial constraints.\footnote{A posinomial constraint is defined as $\sum_{k'=1}^{\widetilde{K}}c_{k'}\prod_{k=1}^{K} x_k^{b_{k'k}}$, where $\{x_k\}$ is the set of optimization variables, $c_{k'} > 0 \forall k'$, and $b_{k'k}$ are real numbers.}  Additionally, the feasible domain of Problem~\eqref{Problem:MaxMinQoS} is convex, so it is a quasi-concave problem as stated in the theorem.
%\vspace*{-0.4cm}
\subsection{Proof of Theorem~\ref{Theorem:Bisection}} \label{Appendix:Bisection}
%\vspace*{-0.2cm}
We start the proof by verifying that every $I_k (\pmb{\rho})$ defined in \eqref{eq:Ikrhov1} is a standard interference function.  The positivity property is satisfied since for all $\pmb{\rho} \succeq \mathbf{0}$, it holds that
%	\vspace*{-0.1cm}
\begin{equation} \label{eq:Positivity}
\begin{split}
I_k (\pmb{\rho}) \geq I_k(\mathbf{0}) &\stackrel{(a)}{=} \frac{ \xi_o \mathsf{NO}_k}{\left|\|\bar{\mathbf{g}}_k\|^2 +  p K \mathrm{tr}(\pmb{\Theta}_k)  +  \sum_{m=1}^M \gamma_{mk} \right|^2} \\
& \stackrel{(b)}{\geq} \frac{ \xi_o \min(\sigma_s^2, \sigma_a^2)}{\|\bar{\mathbf{g}}_k\|^2 +  p K \mathrm{tr}(\pmb{\Theta}_k)  +  \sum_{m=1}^M \gamma_{mk} } >0,
\end{split}
%	\vspace*{-0.1cm}
\end{equation}
where $(a)$ is obtained because the noise power $\mathsf{NO}_k$ is independent of the transmit powers; the result in $(b)$ is obtained by using \eqref{eq:NOk} and doing some further manipulations; and $(c)$ is because the ambient noise always exists in the system. Let us denote the two power vectors by $\pmb{\rho}$ and $\tilde{\pmb{\rho}}$ with $\rho_k \geq \tilde{\pmb{\rho}}_k, \forall k$. Then we observe that
%	\vspace*{-0.1cm}
\begin{equation}
I_k(\pmb{\rho}) - I_k(\tilde{\pmb{\rho}}) = \frac{\left(\xi_o(\mathsf{MI}_k (\pmb{\rho}) - \mathsf{MI}_k (\tilde{\pmb{\rho}}) ) \right)}{\left|\|\bar{\mathbf{g}}_k\|^2 +  p K \mathrm{tr}(\pmb{\Theta}_k)  +  \sum\nolimits_{m=1}^M \gamma_{mk} \right|^2} \geq 0,
 %	\vspace*{-0.1cm}
\end{equation}
which indicates that $I_k(\pmb{\rho}) \geq I_k(\tilde{\pmb{\rho}})$ and therefore the monotonicity property holds true. For a given constant value $\alpha >1$, we observe the scalability property as follows
%	\vspace*{-0.1cm}
\begin{equation} \label{eq:SP}
\begin{split}
\alpha I_k(\pmb{\rho}) &= \frac{\alpha \xi_o\mathsf{MI}_k (\pmb{\rho}) + \alpha \xi_o \mathsf{NO}_k}{\left|\|\bar{\mathbf{g}}_k\|^2 +  p K \mathrm{tr}(\pmb{\Theta}_k)  +  \sum_{m=1}^M \gamma_{mk} \right|^2} \\
&= I_k( \alpha \pmb{\rho}) + \frac{ (\alpha - 1) \xi_o \mathsf{NO}_k}{\left|\|\bar{\mathbf{g}}_k\|^2 +  p K \mathrm{tr}(\pmb{\Theta}_k)  +  \sum_{m=1}^M \gamma_{mk} \right|^2}.
\end{split}
%	\vspace*{-0.1cm}
\end{equation}
The second equality of \eqref{eq:SP} verifies the scalability property, since its second part is nonnegative. Combining \eqref{eq:Positivity}--\eqref{eq:SP}, $I_k (\pmb{\rho}), \forall k,$ represent standard interference functions. Based on \cite[Theorem~$2$]{Yates1995a} and \cite[Theorem~$3$]{van2021uplink}, the convexity of problem~\eqref{Problem:TotalTransmitPower} ensures that for an initial set of data powers, we can exploit the alternating optimization approach for updating the data power of user~$k$ with the standard interference function in an iterative manner. The proposed algorithm converges to a fixed point that is the global optimum of  Problem~\eqref{Problem:TotalTransmitPower}. 

We now derive an upper bound on the SINR values that makes Problem~\eqref{Problem:TotalTransmitPower} infeasible. This upper bound can be concretely defined by solving the following problem:
%	\vspace*{-0.1cm}
\begin{equation} \label{eq:xiupbound}
\xi_o^{\mathrm{up}}  = \underset{k}{\min} \, \sup \, \mathrm{SINR}_k.
%	\vspace*{-0.1cm}
\end{equation}
To solve \eqref{eq:xiupbound}, we use the closed-form expression in \eqref{eq:ClosedSINR} to express the maximal SINR value of user~$k$ as follows
%	\vspace*{-0.1cm}
\begin{equation}
\begin{split}
\mathrm{SINR}_k &\stackrel{(a)}{\leq} \frac{\rho_k \left(\|\bar{\mathbf{g}}_k\|^2 +  p K \mathrm{tr}(\pmb{\Theta}_k)  +  \sum_{m=1}^M \gamma_{mk} \right)^2}{\sigma_s^2 \|\bar{\mathbf{g}}_k\|^2 +  p K \sigma_s^2 \mathrm{tr}( \pmb{\Theta}_k ) + \sigma_a^2 \sum\nolimits_{m=1}^M  \gamma_{mk}} \\
&\stackrel{(b)}{\leq} \frac{P_{\max,k} \left(\|\bar{\mathbf{g}}_k\|^2 +  p K \mathrm{tr}(\pmb{\Theta}_k)  +  \sum_{m=1}^M \gamma_{mk} \right)^2}{\sigma_s^2 \|\bar{\mathbf{g}}_k\|^2 +  p K \sigma_s^2 \mathrm{tr}( \pmb{\Theta}_k ) + \sigma_a^2 \sum\nolimits_{m=1}^M  \gamma_{mk}},
\end{split}
%	\vspace*{-0.1cm}
\end{equation}
where $(a)$ is obtained by ignoring the mutual interference from the other users, and the equality holds as the system serves user~$k$ only; $(b)$ is bounded by the limited power budget, i.e., $\rho_k \leq P_{\max,k}$.  In contrast, if the achievable SINR value of an arbitrary user is lower than its requirement at any given value $\xi_o$, the congestion appears and we can easily detect it by testing the condition \eqref{eq:IkPmaxk}. Thus, the proof is complete.
%\vspace*{-0.4cm}
\subsection{Proof of Theorem~\ref{Theorem:SoftRemoval}} \label{Appendix:SoftRemoval}
%\vspace*{-0.2cm}
%For the sake of simplicity, the iteration index is ignored. 
We start the proof by testing that each $f_k(\pmb{\rho}(n-1))$ in \eqref{eq:rhokn} is a two-sided scalable function as shown in Definition~\ref{Def:TSSF}. When the price of the standard interference function does not exceed the maximum data power, i.e., $\tilde{I}_k(\pmb{\rho}(n-1)) \leq P_{\max,k}, \forall k$, it is sufficient to prove that $\tilde{I}_k(\pmb{\rho}(n-1))$ is a two-sided scalable function. Since $\tilde{I}_k(\pmb{\rho}(n-1))$ is a standard interference function, for any $\alpha >1$ and $\alpha^{-1} \pmb{\rho}(n-1) \preceq \tilde{\pmb{\rho}}(n-1) \preceq \alpha \pmb{\rho}(n-1)$, the following property holds true
%\vspace*{-0.1cm}
\begin{equation} \label{eq:Itildekv1}
\tilde{I}_k(\pmb{\rho}(n-1)) \stackrel{(a)}{<} \tilde{I}_k(\alpha \tilde{\pmb{\rho}}(n-1)) \stackrel{(b)}{<} \alpha \tilde{I}_k( \tilde{\pmb{\rho}}(n-1)),
%\vspace*{-0.1cm}
\end{equation}
where $(a)$ is obtained by exploiting the monotonicity property and $(b)$ is obtained by using the scalability property for the vector $\alpha \pmb{\rho}(n-1)$. From the result in \eqref{eq:Itildekv1}, we have
%\vspace*{-0.1cm}
\begin{equation} \label{eq:Itildekv2}
\tilde{I}_k(\pmb{\rho}(n-1)) <  \alpha \tilde{I}_k( \tilde{\pmb{\rho}}(n-1)),
%\vspace*{-0.1cm}
\end{equation}
which leads to the following observation upon dividing both sides of \eqref{eq:Itildekv2} by the constant $\alpha$ as
%\vspace*{-0.1cm}
\begin{equation}\label{eq:Itildekv3}
\frac{1}{\alpha} \tilde{I}_k(\pmb{\rho}(n-1)) <  \tilde{I}_k( \tilde{\pmb{\rho}}(n-1)).
%\vspace*{-0.1cm}
\end{equation}
By exploiting similar steps associated with the monotonicity and scalability properties for $\tilde{\pmb{\rho}}(n-1) \preceq \alpha \pmb{\rho}(n-1)$, a series of inequalities are further formulated as
%\vspace*{-0.1cm}
\begin{equation} \label{eq:Itildekv4}
\tilde{I}_k(\tilde{\pmb{\rho}}(n-1)) < \tilde{I}_k(\alpha \pmb{\rho}(n-1)) < \alpha  \tilde{I}_k(\pmb{\rho}(n-1)),
%\vspace*{-0.1cm}
\end{equation} 
which leads to the following inequality
%\vspace*{-0.1cm}
\begin{equation} \label{eq:Itildekv5}
\tilde{I}_k(\tilde{\pmb{\rho}}(n-1)) <\alpha  \tilde{I}_k(\pmb{\rho}(n-1)).
%\vspace*{-0.1cm}
\end{equation}
Combining the results in \eqref{eq:Itildekv3} and \eqref{eq:Itildekv5}, we obtain
%\vspace*{-0.1cm}
\begin{equation}\label{eq:Itildekv6}
\frac{1}{\alpha} \tilde{I}_k(\pmb{\rho}(n-1)) <  \tilde{I}_k( \tilde{\pmb{\rho}}(n-1)) < \alpha  \tilde{I}_k(\pmb{\rho}(n-1)),
%\vspace*{-0.1cm}
\end{equation}
which verifies that $f_k(\pmb{\rho}(n-1))$ is a two-sided scalable function when the interference price of  user~$k$ is lower than or equal to the maximum data power. In order to prove that $f_k(\pmb{\rho}_k(n-1))$ is a two-sided scalable function when the interference price of user~$k$ exceeds the maximum data power, i.e., $\tilde{I}_{k}(\pmb{\rho}(n-1)) > P_{\max,k}$, we must prove that $P_{\max, k}^2/ \tilde{I}_{k}(\pmb{\rho}(n-1))$ is also a two-sided scalable function. Indeed, multiplying \eqref{eq:Itildekv6} by the soft removal rate $\mu_k$ and then taking its inverse, one obtains
%\vspace*{-0.1cm}
\begin{equation} \label{eq:Itildekv7}
\frac{1}{\alpha \mu_k \tilde{I}_k(\pmb{\rho}(n-1))} < \frac{1}{ \mu_k \tilde{I}_k( \tilde{\pmb{\rho}}(n-1))} < \alpha \frac{1}{ \mu_k \tilde{I}_k(\pmb{\rho}(n-1)) }.
%\vspace*{-0.1cm}
\end{equation}
Next, multiplying \eqref{eq:Itildekv7} by the factor $P_{\max,k}^2$, we obtain the following result
%\vspace*{-0.1cm}
\begin{equation}
\frac{1}{\alpha}\frac{P_{\max,k}^2}{ \mu_k \tilde{I}_k(\pmb{\rho}(n-1))} < \frac{P_{\max,k}^2}{ \mu_k \tilde{I}_k( \tilde{\pmb{\rho}}(n-1))} < \alpha \frac{P_{\max,k}^2}{ \mu_k \tilde{I}_k(\pmb{\rho}(n-1)) },
%\vspace*{-0.1cm}
\end{equation}
which confirms that the inverse of the standard interference function $I_k(\pmb{\rho}(n-1))$ is two-sided scalable as $ \tilde{I}_k(\pmb{\rho}(n-1)) > P_{\max,k}$. Consequently, each $f_k(\pmb{\rho}(n-1))$ is a two-sided scalable function in its  domain. The SINR constraints of problem~\eqref{Problem:TotalV1} are relaxed for unsatisfied users, allowing us to prove the convergence of two-sided scalable functions. The relaxed SINR constraints ensure that the feasible domain is continuous and bounded, so a fixed-point solution exists.  We now define the distance between the pair of vectors $\pmb{\rho}$ and $\tilde{\pmb{\rho}}$ as
%\vspace*{-0.1cm}
\begin{equation}
d(\pmb{\rho}, \tilde{\pmb{\rho}}) = \underset{k}{\max}\left( \left\{ \max\left( \rho_k/\tilde{\rho}_k, \tilde{\rho}_k /\rho_k \right) \right\}  \right).
%\vspace*{-0.1cm}
\end{equation}
We also stack all the functions $f_k(\pmb{\rho}(n))$ into a vector that is $\mathbf{f}(\pmb{\rho}(n)) = [f_1(\pmb{\rho}(n)), \ldots, f_K(\pmb{\rho}(n))]^T \in \mathbb{R}^{K}$ and let $\pmb{\rho}^\ast$ be the fixed point solution. From the initial data power values stated in the theorem, the following chain of inequalities can be constructed
%\vspace*{-0.1cm}
\begin{equation}
\begin{split}
d(\pmb{\rho}(0) , \pmb{\rho}^\ast) & \stackrel{(a)}{>} d(\mathbf{f}(\pmb{\rho}(0)) ,\mathbf{f}( \pmb{\rho}^\ast))  \\
& \stackrel{(b)}{=}  d(\pmb{\rho}(1) , \pmb{\rho}^\ast ) > \ldots  > d(\mathbf{f}(\pmb{\rho}(n)) , \mathbf{f}(\pmb{\rho}^\ast) ) \\
&\stackrel{(c)}{ =} d( \pmb{\rho}(n+1) , \pmb{\rho}^\ast ) \stackrel{(d)}{=}  d( \pmb{\rho}^\ast , \pmb{\rho}^\ast ) = 1,
 \end{split}
%\vspace*{-0.1cm}
\end{equation}
where $(a)$ is obtained by using \cite[Lemma~7]{sung2005generalized}; $(b)$ and $(c)$ are obtained by the data power update in \eqref{eq:rhokn} and by exploiting the fact that $\mathbf{f}(\pmb{\rho}^\ast) = \pmb{\rho}^\ast$; and $(d)$ is obtained by assuming that the convergence holds at iteration~$n+1$. If the convergence holds at the iteration~$n+1$, it should also hold in the next iterations. The proof is concluded.
%\vspace*{-0.4cm} 
\subsection{Proof of Jain's fairness index in \eqref{eq:JainFairnessIndex}}
%\vspace*{-0.2cm}
Relying on Jain's fairness index by its standard form \cite{jain1984quantitative,bui2021robust}, we define the network fairness as
%\vspace*{-0.1cm}
\begin{equation} \label{eq:JainFairnessIndexv1}
	J = \frac{\left( \sum\nolimits_{k \in \mathcal{K}} R_k (\{ \rho_k^\ast \}) /\hat{\xi}_k \right)^2}{K \sum\nolimits_{k \in \mathcal{K}}  R_k (\{ \rho_k^\ast \})^2/ \hat{\xi}_k^2}.
%\vspace*{-0.1cm}
\end{equation} 
By decomposing the available user set $\mathcal{K}$ into the satisfied user set $\mathcal{K}_s$ defined in \eqref{eq:Ks} and the unsatisfied user set defined in \eqref{eq:Ku}, \eqref{eq:JainFairnessIndexv1} is equivalent to
%\vspace*{-0.1cm}
\begin{equation} \label{eq:JainFairnessIndexv2}
J = \frac{\left( \sum\nolimits_{k \in \mathcal{K}_s } R_k (\{ \rho_k^\ast \}) /\hat{\xi}_k +  \sum\nolimits_{k \in \tilde{\mathcal{K}}_u} R_k (\{ \rho_k^\ast \}) /\hat{\xi}_k \right)^2}{  K \sum\nolimits_{k \in \mathcal{K}_s }  R_k (\{ \rho_k^\ast \})^2/ \hat{\xi}_k^2 +  K \sum\nolimits_{k \in \tilde{\mathcal{K}}_u}  R_k (\{ \rho_k^\ast \})^2/ \hat{\xi}_k^2 }.
%\vspace*{-0.1cm}
\end{equation}
From the total transmit power minimization structure in \eqref{Problem:TotalV1}, it holds that $R_k (\{ \rho_k^\ast \}) /\hat{\xi}_k  = 1, k \in \mathcal{K}_s,$ and therefore we obtain
%\vspace*{-0.1cm}
\begin{equation} \label{eq:Ksv1}
\sum\nolimits_{k \in \mathcal{K}_s } R_k (\{ \rho_k^\ast \}) /\hat{\xi}_k = \sum\nolimits_{k \in \mathcal{K}_s}  R_k (\{ \rho_k^\ast \})^2/ \hat{\xi}_k^2 = |\mathcal{K}_s |.
%\vspace*{-0.1cm}
\end{equation}
By substituting \eqref{eq:Ksv1} into \eqref{eq:JainFairnessIndexv2}, the Jain's fairness index is obtained as in \eqref{eq:JainFairnessIndex}. We emphasize that if all users are satisfied by the requested throughput, i.e., $\mathcal{K}_s = \mathcal{K}$, then $J=1$ by utilizing \eqref{eq:JainFairnessIndex} with $|\mathcal{K}_s| = K$ and $\mathcal{K}_u = \emptyset$. If congestion appears in an extreme case, where user~$k$ gets a non-zero  throughput and the remaining users get zero offers, this leads to $J = 1/K$.
%\vspace*{-0.45cm}
%\setstretch{0.94}
\bibliographystyle{IEEEtran}
\bibliography{IEEEabrv,refs}
\begin{IEEEbiographynophoto}
	{Trinh Van Chien} (S'16-M'20) received the B.S. degree in Electronics and Telecommunications from Hanoi University of Science and Technology (HUST), Vietnam, in 2012. He then received the M.S. degree in Electrical and Computer Enginneering from Sungkyunkwan University (SKKU), Korea, in 2014 and the Ph.D. degree in Communication Systems from Link\"oping University (LiU), Sweden, in 2020. He was  a research associate at University of Luxembourg. He is now with the School of Information and Communication Technology (SoICT), Hanoi University of Science and Technology (HUST), Vietnam. His interest lies in convex optimization problems and machine learning applications for wireless communications and image \& video processing. He was an IEEE wireless communications letters exemplary reviewer for 2016, 2017, and 2021. He also received the award of scientific excellence in the first year of the 5Gwireless project funded by European Union Horizon's 2020.
\end{IEEEbiographynophoto}
\begin{IEEEbiographynophoto} 
	{Eva Lagunas} received the M.Sc. and Ph.D. degrees in telecommunications engineering from the Polytechnic University of Catalonia (UPC), Barcelona, Spain, in 2010 and 2014, respectively. She was Research Assistant within the Department of Signal Theory and Communications, UPC, from 2009 to 2013. During the summer of 2009 she was a guest research assistant within the Department of Information Engineering, Pisa, Italy. From November 2011 to May 2012 she held a visiting research appointment at the Center for Advanced Communications (CAC), Villanova University, PA, USA. In 2014, she joined the Interdisciplinary Centre for Security, Reliability and Trust (SnT), University of Luxembourg, where she currently holds a Research Scientist position. Her research interests include terrestrial and satellite system optimization, spectrum sharing, resource management and machine learning.
\end{IEEEbiographynophoto}
\begin{IEEEbiographynophoto}
	{Tiep M. Hoang}
	received the B.Eng. degree from the HCMC University of Technology, Vietnam, in 2012, the M.Eng. degree from Kyung Hee University, South Korea, in 2014, and the Ph.D. degree from the Queen's University of Belfast, United Kingdom, in 2019. From 2020 to 2022, he was a (postdoctoral) Research Fellow with the School of Electronics and Computer Science, the University of Southampton, United Kingdom. Since May 2022, he has been a Postdoctoral Fellow with the Department of Electrical Engineering, the University of Colorado Denver, United States. His current research interests include 5G/6G wireless communications, wireless security and authentication, reconfigurable intelligent surface (RIS), convex optimization, and machine learning.
\end{IEEEbiographynophoto}
\begin{IEEEbiographynophoto}
	{Symeon Chatzinotas} (S'06, M'09, SM'13, F'23) is Full Professor and Head of the SIGCOM Research Group at SnT, University of Luxembourg. He is coordinating the research activities on communications and networking across a group of 80 researchers, acting as a PI for more than 40 projects and main representative for 3GPP, ETSI, DVB. He is currently serving in the editorial board of the IEEE Transactions on Communications, IEEE Open Journal of Vehicular Technology and the International Journal of Satellite Communications and Networking. In the past, he has been a Visiting Professor at the University of Parma, Italy and was involved in numerous R\&D projects for NCSR Demokritos, CERTH Hellas and CCSR, University of Surrey. He was the co-recipient of the 2014 IEEE Distinguished Contributions to Satellite Communications Award and Best Paper Awards at WCNC, 5GWF, EURASIP JWCN, CROWNCOM, ICSSC. He has (co-)authored more than 700 technical papers in refereed international journals, conferences and scientific books.
\end{IEEEbiographynophoto}
\begin{IEEEbiographynophoto} 
	{Bj\"orn Ottersten} (S'87–M'89–SM'99–F'04) received the M.S. degree in electrical engineering and applied physics from Linköping University, Linköping, Sweden, in 1986, and the Ph.D. degree in electrical engineering from Stanford University, Stanford, CA, USA, in 1990. He has held research positions with the Department of Electrical Engineering, Linköping University, the Information Systems Laboratory, Stanford University, the Katholieke Universiteit Leuven, Leuven, Belgium, and the University of Luxembourg, Luxembourg. From 1996 to 1997, he was the Director of Research with ArrayComm, Inc., a start-up in San Jose, CA, USA, based on his patented technology. In 1991, he was appointed Professor of signal processing with the Royal Institute of Technology (KTH), Stockholm, Sweden. Dr. Ottersten has been Head of the Department for Signals, Sensors, and Systems, KTH, and Dean of the School of Electrical Engineering, KTH. He is currently the Director for the Interdisciplinary Centre for Security, Reliability and Trust, University of Luxembourg. He is a recipient of the IEEE Signal Processing Society Technical Achievement Award, the EURASIP Group Technical Achievement Award, and the European Research Council advanced research grant twice. He has co-authored journal papers that received the IEEE Signal Processing Society Best Paper Award in 1993, 2001, 2006, 2013, and 2019, and 8 IEEE conference papers best paper awards. He has been a board member of IEEE Signal Processing Society, the Swedish Research Council and currently serves of the boards of EURASIP and the Swedish Foundation for Strategic Research. Dr. Ottersten has served as Editor in Chief of EURASIP Signal Processing, and acted on the editorial boards of IEEE Transactions on Signal Processing, IEEE Signal Processing Magazine, IEEE Open Journal for Signal Processing, EURASIP Journal of Advances in Signal Processing and Foundations and Trends in Signal Processing. He is a fellow of EURASIP. 
\end{IEEEbiographynophoto}
\begin{IEEEbiographynophoto} 
	{Lajos Hanzo} (\url{http://www-mobile.ecs.soton.ac.uk}, \url{https://en.wikipedia.org/wiki/Lajos_Hanzo}) is a Fellow of the Royal Academy of Engineering, FIEEE, FIET, Fellow of EURASIP and a Foreign Member of the Hungarian Academy of Sciences. He published prolifically at IEEE Xplore and 19 Wiley-IEEE Press monographs. He was bestowed upon the IEEE Eric Sumner Technical Field Award.
\end{IEEEbiographynophoto}
%\vspace*{-0.5cm}
\end{document}